\documentclass[a4paper]{article}
\pdfoutput=1
\usepackage[dvipsnames]{xcolor} 
\definecolor{linkcolor}{rgb}{0,0.1,0}
\usepackage[backref=page,
             colorlinks,
             linkcolor=blue,
             citecolor=linkcolor,
             linkcolor=linkcolor,
             urlcolor=linkcolor
             ]{hyperref}
\usepackage{amssymb,amsmath}
\usepackage{amsthm}
\usepackage[T1]{fontenc}
\usepackage[english]{babel}
\usepackage[a4paper]{geometry}
\usepackage{microtype}
\usepackage{authblk}

\usepackage{xparse}
\usepackage{bbm}
\usepackage{stmaryrd}
\usepackage{mathpartir}
\usepackage{tikz-cd}
\usepackage{listings}
\usepackage{mathtools}

\usepackage{authblk}
\title{Guarded Cubical Type Theory}

\author[1]{Lars Birkedal}
\author[1]{Ale\v{s} Bizjak}
\author[1]{Ranald Clouston}
\author[1]{Hans~Bugge~Grathwohl}
\author[1]{Bas~Spitters}
\author[2]{Andrea Vezzosi}
\affil[1]{Department of Computer Science, Aarhus University, Denmark}
\affil[2]{Department of Computer Science and Engineering, Chalmers University of Technology, Sweden}

\theoremstyle{plain}
\newtheorem{theorem}{Theorem}[section]
\newtheorem{lemma}[theorem]{Lemma}
\newtheorem{corollary}[theorem]{Corollary}
\theoremstyle{definition}
\newtheorem{definition}[theorem]{Definition}
\newtheorem{example}[theorem]{Example}
\theoremstyle{remark}

\newenvironment{diagram}{\begin{tikzcd}[sep=large]}{\end{tikzcd}}

\DeclareDocumentCommand{\later}{ o m }{
  \IfNoValueTF{#1}
  {\mathord{\triangleright}#2}
  {\mathord{\triangleright}#1 . #2}}

\DeclareDocumentCommand{\pure}{ o m }{
  \IfNoValueTF{#1}
  {\term{next}#2}
  {\term{next} #1 . \, #2}}

\DeclareDocumentCommand{\gstream}{ o m }
{\IfNoValueTF{#1}
  {\ensuremath{\term{Str}_{#2}}}
  {\ensuremath{\term{Str}^{#1}_{#2}}}}

\DeclareDocumentCommand{\eqjudg}{ m m m o }{
    \IfNoValueTF{#4}
    {\ensuremath{#1 \vdash #2 = #3}}
    {\ensuremath{#1 \vdash #2 = #3 : #4}}
}

\DeclareDocumentCommand{\dfix}{ o m m }{
    \IfNoValueTF{#1}
    {\dfixEmp #2 . #3}
    {\dfixEmp^{#1} #2 . #3}
}

\DeclareDocumentCommand{\fix}{ o m m }{
    \IfNoValueTF{#1}
    {\fixEmp #2 . #3}
    {\fixEmp^{#1} #2 . #3}
}

\newcommand{\LfixEmp}{\operatorname{fix}}
\DeclareDocumentCommand{\Lfix}{ m m }{
  \LfixEmp #1 . #2
}

\newcommand{\CC}{\mathbb{C}}
\newcommand{\DD}{\mathbb{D}}
\newcommand{\Face}{\ensuremath{\mathbbm{F}}}
\newcommand{\Glue}{\term{Glue}}
\newcommand{\Id}{\term{Id}}
\newcommand{\Ineg}[1]{\ensuremath{1-#1}}
\newcommand{\I}{\ensuremath{\mathbbm{I}}}
\newcommand{\LCompTy}[2]{\Phi(#1;#2)}
\newcommand{\LElCompEmp}{\operatorname{Comp}}
\newcommand{\LElComp}[1]{\LElCompEmp(#1)}
\newcommand{\LElEmp}{\operatorname{El}}
\newcommand{\LEl}[1]{\LElEmp(#1)}
\newcommand{\LFace}{\mathbbm{F}}
\newcommand{\LGlueEmp}{\operatorname{Glue}}
\newcommand{\LGlue}[4]{\LGlueEmp \, \left[#1 \mapsto (#2,#3)\right] \, #4}
\newcommand{\LglueEmp}{\operatorname{glue}}
\newcommand{\Lglue}[3]{\LglueEmp \, \left[#1 \mapsto #2\right] \, #3}
\newcommand{\Lunglue}{\operatorname{unglue}}
\newcommand{\LIdop}{\operatorname{Id}}
\newcommand{\LId}[3]{\LIdop_{#1} (#2, #3 )}

\newcommand{\LUf}{\LU_{f}}

\newcommand{\LU}{\mathcal{U}}
\newcommand{\Lfaceto}[2]{\ensuremath{\Lface{#2}\to #1}}
\newcommand{\Lface}[1]{\left[#1\right]}
\newcommand{\Lhastype}[3]{\ensuremath{#1 \vdash #2 : #3}}
\newcommand{\Llater}{\operatorname{\triangleright}}
\newcommand{\Lwftype}[2]{\ensuremath{#1 \vdash #2}}
\newcommand{\Nat}{\term{N}}
\newcommand{\OneF}{\ensuremath{1_{\Face}}}
\newcommand{\Path}{\term{Path}}
\newcommand{\UG}{\mathfrak{U}}
\newcommand{\U}{\term{U}}
\newcommand{\Ycom}{\term{Y}}
\newcommand{\ZeroF}{\ensuremath{0_{\Face}}}

\newcommand{\abs}[2]{\langle #1 \rangle #2}
\newcommand{\app}{\circledast}
\newcommand{\bfc}{\mathbf{c}}
\newcommand{\bff}{\mathbf{f}}
\newcommand{\bfl}{\boldsymbol\ell}
\newcommand{\bisim}{\term{bisim}}
\newcommand{\bnfeq}{::=}
\newcommand{\cL}{\ensuremath{\mathcal{L}}}
\newcommand{\code}[1]{\widehat{#1}}
\newcommand{\comp}[4]{\term{comp}^#1\,#2~#3~#4}
\newcommand{\ctt}{\ensuremath{\mathsf{CTT}}}
\newcommand{\cube}{\mathcal{C}}
\newcommand{\dM}[1]{\mathbf{DM}\left(#1\right)}
\newcommand{\defeq}{\triangleq}
\newcommand{\den}[1]{\left\llbracket #1 \right\rrbracket}
\newcommand{\dfixEmp}{\term{dfix}}
\newcommand{\dsubst}[3]{\ensuremath{\vdash #1 : #2\rightarrowtriangle #3}}
\newcommand{\Elem}{\ensuremath{\term{El}}}
\newcommand{\Equiv}{\ensuremath{\term{Equiv}}}
\newcommand{\finset}{\operatorname{Fin}}
\newcommand{\fixEmp}{\term{fix}}
\newcommand{\fold}{\term{fold}}
\newcommand{\gctt}{\ensuremath{\mathsf{GCTT}}}
\newcommand{\gdtt}{\ensuremath{\mathsf{GDTT}}}
\newcommand{\hastype}[3]{\ensuremath{#1 \vdash #2 : #3}}
\newcommand{\hd}{\term{hd}}
\newcommand{\Hom}[3]{\ensuremath{\mathbf{Hom}_{#1}\left(#2,#3\right)}}
\newcommand{\hrt}[1]{\left[ #1 \right]}
\newcommand{\IDD}{\I^{\DD}}
\newcommand{\II}{\I}
\newcommand{\interpret}[1]{\llbracket #1 \rrbracket}
\newcommand{\isetsep}{\;\ifnum\currentgrouptype=16 \middle\fi|\;}

\newcommand{\Iw}{\I^{\omega}}
\newcommand{\join}{\vee}
\newcommand{\KK}{\mathbb{K}}
\newcommand{\latercode}{\code{\triangleright}}
\DeclareDocumentCommand{\Leq}{ m m m o }{
    \IfNoValueTF{#4}
    {\ensuremath{#1 \vdash #2 = #3}}
    {\ensuremath{#1 \vdash #2 = #3 : #4}}
}
\newcommand{\LEq}[2]{{#1 = #2}}
\newcommand{\LFaceDD}{\LFace^{\DD}}
\newcommand{\LFillTy}[2]{\Psi(#1,#2)}
\newcommand{\LNat}{\mathbb{N}}
\newcommand{\Lnext}{\operatorname{next}}
\newcommand{\LPathop}{\operatorname{Path}}
\newcommand{\LPath}[3]{\LPathop_{#1} \, #2 ~ #3}
\newcommand{\Lrefl}{\operatorname{refl}}
\newcommand{\LTrue}{\operatorname{\top}}
\newcommand{\LUDD}{\mathcal{U}^\DD}
\newcommand{\LUfDD}{\LU_{f}^\DD}
\newcommand{\Lwfcxt}[1]{\ensuremath{#1 \vdash}}
\newcommand{\meet}{\wedge}
\newcommand{\natrec}{\term{natrec}}
\newcommand{\op}{^{\operatorname{op}}}
\newcommand{\ott}{\ensuremath{\mathsf{OTT}}}
\newcommand{\pathlambda}[1]{\mathop{\ensuremath{\langle #1 \rangle}}}
\newcommand{\prev}{\operatorname{prev}}
\newcommand{\Prop}{\Omega}
\newcommand{\pullbacktip}{\arrow[dr, phantom, "\scalebox{2}{\ensuremath{\lrcorner}}", very near start]}
\newcommand{\psh}[1]{\ensuremath{\widehat{#1}}}
\newcommand{\restr}[2]{#1_{\restriction_{#2}}}
\newcommand{\Sets}{\operatorname{Set}}
\newcommand{\subst}[2]{[#1/#2]}
\newcommand{\term}[1]{\ensuremath{\operatorname{\mathsf{#1}}}}
\newcommand{\tl}{\term{tl}}
\newcommand{\Tl}{\ensuremath{\mathcal{T}}}
\newcommand{\Tm}{\operatorname{Tm}}
\newcommand{\totcube}{{\widehat{\cube \times \omega}}}
\newcommand{\two}{\ensuremath{\mathbbm{2}}}
\newcommand{\Ty}{\operatorname{Ty}}
\newcommand{\unfold}{\term{unfold}}
\newcommand{\vrt}[1]{\left[\begin{array}{l} #1 \end{array}\right]}
\newcommand{\wfcxt}[1]{\ensuremath{#1 \vdash}}
\newcommand{\wftype}[2]{\ensuremath{#1 \vdash #2}}
\newcommand{\zero}{\term{0}}

\renewcommand{\phi}{\varphi}
\renewcommand{\succ}{\term{s}}

\theoremstyle{plain}
\newtheorem{proposition}[theorem]{Proposition}

\newcommand{\gctturl}{\url{http://github.com/hansbugge/cubicaltt/tree/gcubical}}

\theoremstyle{definition}
\newtheorem{assumption}{Assumption}

\makeatletter
\newcommand{\pushright}[1]{\ifmeasuring@#1\else\omit\hfill$\displaystyle#1$\fi\ignorespaces}
\makeatother

\begin{document}

\maketitle

\begin{abstract}
This paper improves the treatment of equality in guarded dependent type theory ($\gdtt$), by combining it with cubical type theory ($\ctt$).
$\gdtt$ is an extensional type theory with guarded recursive types, which are useful for building models of program logics, and for programming and reasoning with coinductive types.
We wish to implement $\gdtt$ with decidable type checking, while still supporting non-trivial equality proofs that reason about the extensions of guarded recursive constructions.
$\ctt$ is a variation of Martin-L\"of type theory in which the identity type is replaced by abstract paths between terms.
$\ctt$ provides a computational interpretation of functional extensionality, enjoys canonicity for the natural numbers type, and is conjectured to support decidable type-checking.
Our new type theory, guarded cubical type theory ($\gctt$), provides a computational interpretation of extensionality for guarded recursive types.
This further expands the foundations of $\ctt$ as a basis for formalisation in mathematics and computer science.
We present examples to demonstrate the expressivity of our type theory, all of which have
been checked using a prototype type-checker implementation.
We show that $\ctt$ can be given semantics in presheaves on $\cube \times \DD$, where $\cube$ is the cube category,
and $\DD$ is any small category with an initial object. We then show that the category of presheaves on
$\cube \times \omega$ provides semantics for $\gctt$.
\end{abstract}


\section{Introduction}
\label{sec:intro}

Guarded recursion is a technique for defining and reasoning about infinite objects. Its applications include the definition of
productive operations on data structures more commonly defined via coinduction, such as streams, and the construction
of models of program logics for modern programming languages with features such as higher-order store and
concurrency~\cite{Birkedal:Step}. This is done via the type-former $\later$, called `later', which distinguishes
data which is available immediately from data only available after some computation, such as the unfolding of a
fixed-point. For example, guarded recursive streams are defined by the equation
\[
  \gstream{A} \;=\; A \times \later{\gstream{A}}
\]
rather than the more standard $\gstream{A} = A \times\gstream{A}$, to specify that the head is available now but
the tail only later. The type for guarded fixed-point combinators is then $(\later A\to A)\to A$, rather than the logically inconsistent
$(A\to A)\to A$, disallowing unproductive definitions such as taking the fixed-point of the identity function.

Clouston et al.~\cite{Clouston:Guarded} developed guarded recursive types in a simply-typed setting, following
earlier work~\cite{Nakano:Modality,Atkey:Productive,Abel:Formalized}, with semantics in the presheaf category
$\psh{\omega}$ known as the \emph{topos of trees}, and also presented a logic for reasoning about programs with
guarded recursion.
For large examples, such as models of program logics, we would like to be able to formalise such reasoning. A major
approach to formalisation is via \emph{dependent types}, used for example in the proof assistants Coq~\cite{Coq:manual} and Agda~\cite{Norell:thesis}.
Bizjak et al.~\cite{Bizjak-et-al:GDTT}, following earlier
work~\cite{Birkedal+:topos-of-trees,Mogelberg:tt-productive-coprogramming}, introduced guarded dependent type theory
($\gdtt$), integrating the $\later$ type-former into a dependently typed calculus, and supporting the definition of guarded
recursive types as fixed-points of functions on universes, and guarded recursive operations on these types.

We wish to formalise non-trivial theorems about equality between guarded recursive constructions, but such
arguments often cannot be accommodated within \emph{intensional} Martin-L\"{o}f type theory. For example, we may need
to be able to reason about the extensions of streams in order to prove the equality of different stream functions. Hence
$\gdtt$ includes an equality reflection rule, which is well known to make type checking undecidable.
This problem is close to well-known problems with functional
extensionality~\cite[Section~3.1.3]{Hofmann:Extensional}, and indeed this analogy can be developed. Just as
functional extensionality involves mapping terms of type $(x:A)\to\Id B\, (fx)\, (gx)$ to proofs of $\Id\,(A\to B)\, f\, g$, extensionality
for guarded recursion requires an extensionality principle for later types, namely the ability to map terms of type
$\later\Id A\,t\,u$ to proofs of
$\Id\,(\later A)\,(\pure{t})\,(\pure{u})$, where $\pure{}$ is the constructor for $\later$. These types are isomorphic in the topos of trees, and so in $\gdtt$ their equality was asserted as an axiom.
But in a calculus without equality reflection we cannot merely assert such axioms without losing canonicity.

\emph{Cubical type theory} ($\ctt$)~\cite{Cubical}, for which we give a brief introduction in Section~\ref{sec:ctt}, is a new
type theory with a computational interpretation of functional
extensionality but without equality reflection, and hence is a candidate for extension with guarded recursion, so that we may
formalise our arguments without incurring the disadvantages of fully extensional identity
types. $\ctt$ was developed primarily to provide a computational interpretation of Voevodsky's
univalence axiom in Homotopy Type Theory~\cite{hottbook}. The most important novelty of $\ctt$ is the replacement of inductively defined identity types by \emph{paths}, which can be seen as maps from an abstract interval, and are introduced and eliminated much like functions. $\ctt$ can be extended with identity types which model all rules of  intensional equality in Martin-L\"of type
theory~\cite[Sec. 9.1]{Cubical}, but these are logically equivalent to path types, and in our paper it suffices to work with
path types only.
$\ctt$ has sound denotational semantics in (fibrations in) \emph{cubical
  sets}, a presheaf category that is used to model homotopy types. $\ctt$ enjoys
canonicity for terms of natural number type~\cite{Huber:canonicity} and is conjectured to have decidability of
type-checking. Moreover, a type-checker has been implemented%
\footnote{\url{https://github.com/mortberg/cubicaltt}}.

In Section~\ref{sec:type-theory-examples} of this paper we propose \emph{guarded cubical type theory} ($\gctt$), a
combination of the two type theories%
\footnote{with the exception of the \emph{clock quantification} of $\gdtt$, which we leave to future work.}
which supports non-trivial proofs about guarded recursive types via path equality, while retaining the potential for good
syntactic properties such as canonicity for base types and decidable type-checking. In particular, just as a term can be defined in $\ctt$ to
witness functional extensionality, a term can be defined in $\gctt$ to witness extensionality for later types.
Further, we use elements of the interval of $\ctt$ to annotate fixed-points, and hence control their unfoldings.
This ensures that fixed-points are path equal, but not judgementally equal, to their unfoldings, and hence prevents
infinite unfoldings, an obvious source of non-termination in any calculus with infinite constructions.
The resulting calculus is shown via examples to be useful for reasoning about guarded recursive operations;
we also view it as potentially significant from the point of view of $\ctt$, extending its expressivity as a basis for formalisation.

In Section~\ref{sec:semantics} we give semantics to
this type theory via the presheaf category over the product of the categories used to define semantics for $\gdtt$ and $\ctt$.
Defining semantics in this new category is non-trivial because we must check that all novel features of the two type
theories can still be soundly interpreted.
To achieve this we first define, in Section~\ref{sec:internal}, an extension of dependent predicate logic in which
the constructs of $\ctt$ may be interpreted, then show that this logic may be interpreted in a certain class of
presheaf categories, including our intended category.
We then show that this category also interprets the constructs of $\gdtt$.
In particular we must ensure that the `later' functor $\later$, which models
the type-former of the same name, preserves the (Kan) composition operations which are central to the cubical model.
In the conference version of this paper~\cite{CSL} the development of the semantics was presented only very briefly for
space reasons; the technical appendix of that paper is integrated into the text of this paper, and forms the bulk of this
paper's contribution.

Moreover, we have implemented a prototype type-checker for this extended type theory%
\footnote{\gctturl},
extending the implemented type-checker for $\ctt$,
which provides confidence in our type theory's syntactic properties. All constructions using the type theory
$\gctt$ presented in this paper,
and many others, have been formalised in this type-checker.


\section{Cubical Type Theory}
\label{sec:ctt}

This section gives a brief overview of \emph{cubical type theory}  ($\ctt$)%
\footnote{\url{http://www.cse.chalmers.se/~coquand/selfcontained.pdf} provides a concise presentation of
$\ctt$.};
for full details we refer to Cohen et al.~\cite{Cubical}.

We start with a standard dependent type theory with $\Pi$, $\Sigma$, natural numbers, and a Russell-style universe, but
without identity types:
\[
  \begin{array}{lcl@{\hspace{.2\linewidth}}l}
    \Gamma, \Delta & \bnfeq & () ~ | ~ \Gamma, x : A & \text{Contexts} \\[1ex]
    t,u,A,B  & \bnfeq & x &\text{Variables} \\
    & | & \lambda x : A . t ~|~ t\, u ~|~ (x : A) \to B &\text{$\Pi$-types} \\
    & | & (t,u) ~|~ t.1 ~|~ t.2 ~|~ (x:A) \times B &\text{$\Sigma$-types} \\
    & | & \zero ~|~ \succ t ~|~ \natrec t \, u ~|~ \Nat & \text{Natural numbers} \\
    & | & \U & \text{Universe}
  \end{array}
\]
We adhere to the usual conventions of considering terms and types up to $\alpha$-equality, and writing $A \to B$, respectively
$A \times B$, for non-dependent $\Pi$- and $\Sigma$-types. We use the symbol `$=$' for judgemental equality.

$\ctt$ extends this basic type theory with the constructs below:
\[
  \begin{array}{lcl@{\hspace{.12\linewidth}}l}
    r,s & \bnfeq & 0 ~|~ 1 ~|~ i ~|~ \Ineg{r} ~|~ r \meet s ~|~ r \join s & \text{The interval, $\I$} \\[1ex]
    \phi,\psi & \bnfeq & \ZeroF ~|~ \OneF ~|~ (i=0) ~|~ (i=1) ~|~ \phi \meet \psi ~|~ \phi \join \psi
      & \text{The face lattice, $\Face$} \\[1ex]
    \Gamma, \Delta & \bnfeq & \cdots ~|~ \Gamma, i:\I ~|~ \Gamma,\phi & \text{Contexts} \\[1ex]
    t,u,A,B  & \bnfeq & \cdots \\
    & | & \pathlambda{i} t ~|~ t\,r ~|~  \Path A ~ t ~ u &\text{Path types} \\
    & | & [ \phi_1~t_1,\ldots,\phi_n~t_n ] &\text{Systems} \\
    & | & \comp{i}{A}{[\phi\mapsto u]}{t} &\text{Compositions} \\
    & | & \Lglue{\phi}{t}{u} ~|~ \Lunglue\,t ~|~ \LGlue{\phi}{B}{t}{A} &\text{Glueing}
  \end{array}
\]
We now briefly discuss these constructs.

The central novelty of $\ctt$ is its treatment of equality.
Instead of the inductively defined identity types of intensional Martin-L\"of type theory~\cite{Martin-Lof-1973}, $\ctt$ has
\emph{paths}.
The paths between two terms $t,u$ of type $A$ form a sort of function space, intuitively that of continuous maps from some interval $\I$ to $A$, with endpoints $t$ and $u$.
Rather than defining the interval $\I$ concretely as the unit interval $[0,1] \subseteq
\mathbb{R}$, it is defined as the \emph{free De Morgan algebra}
on a discrete infinite set of names $\{i, j, k, \dots \}$ with endpoints $0$ and $1$. A De Morgan algebra is a bounded distributive lattice with an involution $\Ineg{\cdot}$ satisfying the De Morgan laws
\begin{align*}
  \Ineg{(i\meet j)} &= (\Ineg{i}) \join (\Ineg{j}), & \Ineg{(i\join j)} &= (\Ineg{i}) \meet (\Ineg{j}).
\end{align*}
The interval $[0,1]  \subseteq \mathbb{R}$, with $\term{min}$, $\term{max}$ and $\Ineg{\cdot}$, is an example of a De Morgan algebra.

The judgement $\hastype{\Gamma}{r}{\I}$ means that $r$ draws its names from $\Gamma$.
Despite this notation, $\I$ is not a first-class type.

Path types and their elements are defined by
the rules in Figure~\ref{fig:path-typing-rules}.
\emph{Path abstraction}, $\pathlambda{i} t$, and \emph{path application}, $t \, r$, are analogous to $\lambda$-abstraction and function application, and support the familiar $\beta$-equality $(\pathlambda{i} t)\, r = t\subst{r}{i}$ and $\eta$-equality $\pathlambda{i} t\, i = t$.
There are two additional judgemental equalities for paths, regarding their endpoints: given $p : \Path A ~ t ~ u$ we have $p\, 0 = t$ and $p\, 1 = u$.

\begin{figure}[t]
  \begin{mathpar}
    \inferrule{%
      \wftype{\Gamma}{A} \\
      \hastype{\Gamma}{t}{A} \\
      \hastype{\Gamma}{u}{A}
    }{%
      \wftype{\Gamma}{\Path A ~ t ~ u}}
  \end{mathpar}
  \begin{mathpar}
    \inferrule{%
      \wftype{\Gamma}{A} \\
      \hastype{\Gamma, i:\I}{t}{A}
    }{%
      \hastype{\Gamma}{\pathlambda{i} t}{\Path A ~t\subst{0}{i} ~ t\subst{1}{i}}}
    \and
    \inferrule{%
      \hastype{\Gamma}{t}{\Path A ~ u ~ s} \\
      \hastype{\Gamma}{r}{\I}
    }{%
      \hastype{\Gamma}{t\, r}{A}}
  \end{mathpar}
  \caption{Typing rules for path types.}
  \label{fig:path-typing-rules}
\end{figure}

Paths provide a notion of identity which is more extensional than
that of intensional Martin-L\"of identity types, as exemplified by the proof term for functional extensionality:\label{funext}
\[
  \term{funext} \, f\, g \defeq \lambda p . \pathlambda{i} \lambda x .\, p \, x \, i
  ~ : ~
  \left((x : A) \to \Path B ~ (f\, x) ~ (g\, x)\right) \to \Path~ (A\to B) ~ f ~ g.
\]

The rules above suffice to ensure that path equality is reflexive, symmetric, and a congruence, but we also need it to be
transitive and, where the underlying type is the universe, to support a notion of transport.
This is done via \emph{(Kan) composition operations}.

To define these we need the \emph{face lattice}, $\Face$, defined as the free distributive lattice on the symbols $(i=0)$ and $(i=1)$ for all names $i$, quotiented by the relation $(i=0) \meet (i=1) = \ZeroF$.
As with the interval, $\Face$ is not a first-class type, but the judgement $\hastype{\Gamma}{\phi}{\Face}$ asserts
that $\phi$ draws its names from $\Gamma$.
We also have the judgement $\eqjudg{\Gamma}{\phi}{\psi}[\Face]$ which asserts the equality of $\phi$ and $\psi$ in the
face lattice. Contexts can be restricted by elements of $\Face$.
Such a restriction affects equality judgements so that, for example, $\eqjudg{\Gamma,\phi}{\psi_1}{\psi_2}[\Face]$ is
equivalent to $\eqjudg{\Gamma}{\phi \meet \psi_1}{\phi \meet \psi_2}[\Face]$

We write $\hastype{\Gamma}{t}{A[\phi \mapsto u]}$ as an abbreviation for the two judgements $\hastype{\Gamma}{t}{A}$
and $\eqjudg{\Gamma,\phi}{t}{u}[A]$, noting the restriction with $\phi$ in the equality judgement.
Now the composition operator is defined by the typing and equality rule
\begin{mathpar}
  \inferrule{%
    \hastype{\Gamma}{\phi}{\Face} \\
    \wftype{\Gamma, i:\I}{A} \\
    \hastype{\Gamma, \phi, i : \I}{u}{A} \\
    \hastype{\Gamma}{a_0}{A\subst{0}{i}[\phi\mapsto u\subst{0}{i}]}
  }{%
    \hastype{\Gamma}{\comp{i}{A}{[\phi\mapsto u]}{a_0}}{A\subst{1}{i}[\phi\mapsto u\subst{1}{i}]}
  }.
\end{mathpar}
There are further equations for composition that depend on the type $A$ they are applied to; we omit these from this short
overview.

A simple use of composition is to implement the transport operation for $\Path$ types\label{transport-term}
\[
  \term{transp}^i \, A ~ a
  ~\defeq~
  \comp{i}{A}{[\ZeroF \mapsto []]}{a}
  ~:~
  A\subst{1}{i},
\]
where $a$ has type $A\subst{0}{i}$.
The notation $[]$ stands for the empty \emph{system}.  In general a system is a list
of pairs of faces and terms, and it defines an element of a type by giving the individual
components at each face.
Below we present two of the rules for systems; in particular the first rule ensures that for a system to be well-typed, all
cases must be covered, and the components must agree where the faces overlap:
\begin{mathpar}
  \inferrule{%
    \wftype{\Gamma}{A}\\
    \eqjudg{\Gamma}{\phi_1\join\ldots\join\phi_n}{\OneF}[\Face] \\
    \hastype{\Gamma, \phi_i}{t_i}{A} \\
    \eqjudg{\Gamma,\phi_i\meet\phi_j}{t_i}{t_j}[A]\\ i,j=1 \ldots n \\
  }{%
    \hastype{\Gamma}{[ \phi_1~t_1,\ldots,\phi_n~t_n ]}{A}
  }
  \and
  \inferrule{%
    \hastype{\Gamma}{[ \phi_1~t_1,\ldots,\phi_n~t_n ]}{A}\\
    \eqjudg{\Gamma}{\phi_i}{\OneF}[\Face]
  }{%
    \eqjudg{\Gamma}{[ \phi_1~t_1,\ldots,\phi_n~t_n ]}{t_i}[A]
  }
\end{mathpar}
We will write $[ \phi_1\mapsto t_1,\ldots,\phi_n\mapsto t_n ]$ as an abbreviation for
$[\phi_1\join\ldots\join\phi_n \mapsto [ \phi_1~t_1,\ldots,\phi_n~t_n ]]$.

A non-trivial example of the use of systems is the proof that $\Path$ is transitive; given $p\,:\,\Path A~a~b$ and
$q\,:\,\Path A~b~c$ we can define   
\[
\term{transitivity}\,p\,q \defeq \pathlambda{i} \comp{j}{A}{[(i=0) \mapsto a, (i=1) \mapsto q\,j]}{(p \, i)} \, : \, \Path A~a~c.
\]
This builds a path between the appropriate endpoints because we have the equalities $\comp{j}{A}{[\OneF \mapsto a]}{(p \, 0)} = a$ and $\comp{j}{A}{[\OneF \mapsto q\,j]}{(p \, 1)} = q\,1 = c$.

The \emph{glueing} construction~\cite[Sec.~6]{Cubical} is necessary to define the interaction of the universe with
compositions, and hence to provide a computational interpretation of univalence.
It has the following type-formation and typing rules:
\begin{mathpar}
  \inferrule{%
    \wftype{\Gamma}{A} \\
    \wftype{\Gamma, \phi}{T} \\
    \hastype{\Gamma, \phi}{f}{\Equiv\,T\,A}}{%
    \wftype{\Gamma}{\LGlue{\phi}{T}{f}{A}}}
  \and
  \inferrule{%
    \hastype{\Gamma}{b}{\LGlue{\phi}{T}{f}{A}}}{%
    \hastype{\Gamma}{\Lunglue b}{A[\phi \mapsto f\, b]}}
  \and
  \inferrule{%
    \hastype{\Gamma, \phi}{f}{\Equiv\,T\,A} \\
    \hastype{\Gamma, \phi}{t}{T} \\
    \hastype{\Gamma}{a}{A[\phi \mapsto f\, t]}}{%
    \hastype{\Gamma}{\Lglue{\phi}{t}{a}}{\LGlue{\phi}{T}{f}{A}}}
\end{mathpar}
where $\Equiv\,T\,A$ is the type of equivalence between types $T$ and $A$, whose formal definition we omit.
We also have the following equations:
\begin{align*}
  \LGlue{\OneF}{T}{f}{A} &=, T \\
  \Lglue{\OneF}{t}{a} &= t, \\
  \Lglue{\phi}{b}{(\Lunglue b)} &= b, \\
  \Lunglue (\Lglue{\phi}{t}{a}) & = a.
\end{align*}


\section{Guarded Cubical Type Theory}
\label{sec:type-theory-examples}

The section introduces constructs from guarded dependent type theory ($\gdtt$) to $\ctt$, to define guarded cubical type
theory ($\gctt$):
\[
  \begin{array}{lcl@{\hspace{.12\linewidth}}l}
    \xi & \bnfeq & \cdot ~|~ \xi\hrt{x \gets t} & \text{Delayed substitutions} \\[1ex]
    t,u,A,B  & \bnfeq & \cdots \\
    & | & \pure[\xi]{t} ~|~\dfix[r]{x}{t} ~|~ \later[\xi]{A} &\text{Later types}
  \end{array}
\]
recalling that $r$ is an element of the interval.
This section will also present examples that show how $\gctt$ can be used to prove properties of guarded recursive
constructions.

\subsection{Later Types}\label{sec:later}

In Figure~\ref{fig:typing-rules-later} we present the `later' types of guarded dependent type theory
($\gdtt$)~\cite{Bizjak-et-al:GDTT}, with judgemental equalities in Figures~\ref{fig:ty-eq-rules-later}
and~\ref{fig:tm-eq-rules-later}. Note that we do not add any new equation for the interaction of compositions with $\later$:
while $\comp{i}{\later[\xi]{A}}{[\phi \mapsto u]}{t}$ is a valid term which allows us to transport at $\later$ types, 
any extra equation for it would be necessary only if we were to add the `previous' eliminator $\prev$ for $\later$, but this extension (which involves
clock quantifiers) is left to further work.
We delay the presentation of the fixed-point construction until the next subsection.

The typing rules use the \emph{delayed substitutions} of $\gdtt$, as defined in
Figure~\ref{fig:del-substs}. The notation $\Gamma\rightarrowtriangle\Gamma'$ for the delayed
substitution is suggestive for its intended
semantics as $\Gamma\to\later(\Gamma,\Gamma')$.
Delayed substitutions resemble Haskell-style do-notation, or a delayed form of let-binding.
If we have a term $t:\later{A}$, we cannot access its contents `now', but if we are defining a type or term that itself has
some part that is available `later', then this part \emph{should} be able to use the contents of $t$.
Therefore delayed substitutions allow terms of type $\later{A}$ to be unwrapped by $\later$ and $\pure$.
As observed by Bizjak et al.~\cite{Bizjak-et-al:GDTT}, these constructions generalise the \emph{applicative
functor}~\cite{McBride:Applicative} structure of `later' types, by the definitions $\term{pure} \, t \defeq \pure{t}$, and $f \app t \defeq
\pure[\hrt{f' \gets f, t' \gets t}]{f' \, t'}$, and also generalise the $\app$ operation from simple functions to
$\Pi$-types. We here make the new observation that delayed substitutions can express the function $\latercode:\later{\U}
\to\U$, introduced by Birkedal and M{\o}gelberg~\cite{Mogelberg:2013} to express guarded recursive types as fixed-points
on universes, as $\lambda u.\later[[u'\gets u]]{u'}$; see for example the definition of streams in
Section~\ref{sec:streams}.

\begin{figure}
  \begin{mathpar}
    \inferrule{%
      \wfcxt{\Gamma}}{%
      \dsubst{\cdot}{\Gamma}{\cdot}}
    \and
    \inferrule{%
      \dsubst{\xi}{\Gamma}{\Gamma'} \\
      \hastype{\Gamma}{t}{\later[\xi]{A}}}{%
      \dsubst{\xi\hrt{x \gets t}}{\Gamma}{\Gamma', x:A}}
  \end{mathpar}
  \caption{Formation rules for delayed substitutions.}
  \label{fig:del-substs}
\end{figure}

\begin{figure}
  \begin{mathpar}
    \inferrule{%
      \wftype{\Gamma,\Gamma'}{A} \\
      \dsubst{\xi}{\Gamma}{\Gamma'}}{%
      \wftype{\Gamma}{\later[\xi]{A}}}
    \and
    \inferrule{%
      \hastype{\Gamma, \Gamma'}{A}{\U} \\
      \dsubst{\xi}{\Gamma}{\Gamma'}}{%
      \hastype{\Gamma}{\later[\xi]{A}}{\U}}
    \and
    \inferrule{%
      \hastype{\Gamma,\Gamma'}{t}{A} \\
      \dsubst{\xi}{\Gamma}{\Gamma'}}{%
      \hastype{\Gamma}{\pure[\xi]{t}}{\later[\xi]{A}}}
  \end{mathpar}
  \caption{Typing rules for later types.}
  \label{fig:typing-rules-later}
\end{figure}

\begin{figure}
    \begin{mathpar}
    \inferrule{%
      \dsubst{\xi\hrt{x\gets t}}{\Gamma}{\Gamma',x:B} \\
      \wftype{\Gamma,\Gamma'}{A}}{%
      \eqjudg{\Gamma}{\later[\xi\hrt{x \gets t}]{A}}{\later[\xi]{A}}}
    \and
    \inferrule{%
      \dsubst{\xi\hrt{x\gets t,y\gets u}\xi'}{\Gamma}{\Gamma',x:B,y:C,\Gamma''} \\
      \wftype{\Gamma,\Gamma'}{C} \\
      \wftype{\Gamma,\Gamma',x:B,y:C,\Gamma''}{A}}{%
      \eqjudg{\Gamma}{\later[\xi\hrt{x\gets t,y\gets u}\xi']{A}}
      {\later[\xi\hrt{y\gets u,x\gets t}\xi']{A}}}
    \and
    \inferrule{%
      \dsubst{\xi}{\Gamma}{\Gamma'} \\
      \wftype{\Gamma,\Gamma', x:B}{A} \\
      \hastype{\Gamma,\Gamma'}{t}{B}}{%
      \eqjudg{\Gamma}{\later[\xi\hrt{x\gets \pure[\xi]{t}}]{A}}
      {\later[\xi]{A\subst{t}{x}}}}
  \end{mathpar}

  \caption{Type equality rules for later types (congruence and equivalence rules are omitted).}
  \label{fig:ty-eq-rules-later}
\end{figure}

\begin{figure}
    \begin{mathpar}
    \inferrule{%
      \dsubst{\xi\hrt{x\gets t}}{\Gamma}{\Gamma', x:B} \\
      \hastype{\Gamma,\Gamma'}{u}{A} }{%
      \eqjudg{\Gamma}{%
        \pure[\xi\hrt{x\gets t}]{u}}{%
        \pure[\xi]{u}}[%
      \later[\xi]{A}]}
    \and
    \inferrule{%
      \dsubst{\xi\hrt{x\gets t,y\gets u}\xi'}{\Gamma}{\Gamma',x:B,y:C,\Gamma''} \\
      \wftype{\Gamma,\Gamma'}{C} \\
      \hastype{\Gamma,\Gamma',x:B,y:C,\Gamma''}{v}{A} }{%
      \eqjudg{\Gamma}{%
        \pure[\xi\hrt{x\gets t,y\gets u}\xi']{v}}{%
        \pure[\xi\hrt{y\gets u,x\gets t}\xi']{v}}[%
      \later[\xi\hrt{x\gets t,y\gets u}\xi']{A}]}
    \and
    \inferrule{%
      \dsubst{\xi}{\Gamma}{\Gamma'} \\
      \hastype{\Gamma,\Gamma',x:B}{u}{A} \\
      \hastype{\Gamma,\Gamma'}{t}{B}}{%
      \eqjudg{\Gamma}%
      {\pure[\xi\hrt{x \gets \pure[\xi]{t}}]{u}}%
      {\pure[\xi]{u\subst{t}{x}}}%
      [\later[\xi]{A\subst{t}{x}}] }
    \and
    \inferrule{%
      \hastype{\Gamma}{t}{\later[\xi]{A}} }{%
      \eqjudg{\Gamma}%
      {\pure[\xi\hrt{x\gets t}]{x}}%
      {t}%
      [\later[\xi]{A}]}
  \end{mathpar}
  \caption{Term equality rules for later types. We omit congruence and equivalence rules, and the rules for terms of type
  $\U$, which reflect the type equality rules of Figure~\ref{fig:ty-eq-rules-later}.}
  \label{fig:tm-eq-rules-later}
\end{figure}

\begin{example}

In $\gdtt$ it is essential that we can convert terms of type $\later[\xi]{\term{Id}_A \, t~u}$ into terms of type $\term{Id}_{\later[\xi]{A}} \, (\pure[\xi]{t})~(\pure[\xi]{u})$, so that we may perform \emph{L\"ob induction}, the
technique of proof by guarded recursion where we assume $\later{p}$, deduce $p$, and hence may conclude $p$ with no
assumptions.
This is achieved in $\gdtt$ by postulating as an axiom the following judgemental equality:
\begin{equation}\label{eq:Id_and_later}
  \term{Id}_{\later[\xi]{A}} \, (\pure[\xi]{t})~(\pure[\xi]{u})
  \;=\;
  \later[\xi]{\term{Id}_A \, t~u}
\end{equation}
A term from left-to-right of \eqref{eq:Id_and_later} can be defined using the $\term{J}$-eliminator for identity types, but
the more useful direction is right-to-left, as proofs of equality by L\"ob induction involve assuming that we later have an
equality, then converting this into an equality on later types.
In fact with the paths of $\gctt$ we can define a term with the desired type:
\begin{equation}\label{eq:later_ext}
  \lambda p.\abs{i}{\pure[\xi[p'\gets p]]{p'\, i}} \;:\; (\later[\xi]{\Path A\, t\, u})\to
    \Path\,(\later[\xi]{A})\,(\pure[\xi]{t})\,(\pure[\xi]{u}).
\end{equation}
Note the similarity of this term and type with that of $\term{funext}$, for functional extensionality, presented on
page~\pageref{funext}. Indeed we claim that \eqref{eq:later_ext} provides a computational interpretation of extensionality
for later types.
\end{example}

\subsection{Fixed Points}
\label{sec:fix}

In this section we complete the presentation of $\gctt$ by addressing fixed points.
In $\gdtt$ there are fixed-point constructions $\fix{x}{t}$ with the judgemental equality $\fix{x}{t} =
t\subst{\pure{\fix{x}{t}}}{x}$.
In $\gctt$ we want decidable type checking, including decidable judgemental equality, and so
we cannot admit such an unrestricted unfolding rule.
Our solution is that fixed points should not be judgementally equal to their unfoldings, but merely \emph{path} equal.
We achieve this by decorating the fixed-point combinator with an interval element which specifies the position on this path.
The $0$-endpoint of the path is the stuck fixed-point term, while the $1$-endpoint is the same term unfolded once.
However this threatens canonicity for base types: if we allow stuck fixed-points in our calculus, we could have stuck closed terms $\fix[i]{x}{t}$ inhabiting $\Nat$.
To avoid this, we introduce the \emph{delayed} fixed-point combinator $\dfixEmp$, inspired by Sacchini's guarded
unfolding operator~\cite{Sacchini:Well}, which produces a term `later' instead of a term `now'.
Its typing rule, and notion of equality, is given in Figure~\ref{fig:dfix}.
We will write $\fix[r]{x}{t}$ for $t\subst{\dfix[r]{x}{t}}{x}$, $\fix{x}{t}$ for $\fix[0]{x}{t}$, and $\dfix{x}{t}$ for $\dfix[0]{x}{t}$.

\begin{figure}
\begin{mathpar}
  \inferrule{%
    \hastype{\Gamma}{r}{\I} \\
    \hastype{\Gamma, x : \later{A}}{t}{A}
  }{%
    \hastype{\Gamma}{\dfix[r]{x}{t}}{\later{A}}}
  \and
  \inferrule{%
    \hastype{\Gamma, x : \later{A}}{t}{A} }{%
    \eqjudg{\Gamma}%
    {\dfix[1]{x}{t}}%
    {\pure{t\subst{\dfix[0]{x}{t}}{x}}}%
    [\later{A}]}.
\end{mathpar}
  \caption{Typing and equality rules for the delayed fixed-point}
  \label{fig:dfix}
\end{figure}

\begin{lemma}[Canonical unfold lemma]
  \label{prop:unfold-lemma}
  For any term $\hastype{\Gamma, x : \later{A}}{t}{A}$ there is a path between $\fix{x}{t}$ and $t\subst{\pure{\fix{x}{t}}}{x}$, given by the term $\pathlambda{i} \fix[i]{x}{t}$.
\end{lemma}

Transitivity of paths (via compositions) ensures that $\fix{x}{t}$ is path equal to any number of fixed-point unfoldings of itself.

A term $a$ of type $A$ is said to be a \emph{guarded fixed point} of a function $f:\later{A}\to A$ if there is a
path from $a$ to $f(\pure{a})$.

\begin{proposition}[Unique guarded fixed points]
  \label{prop:unique-fix}
  Any guarded fixed-point $a$ of a term $f : \later{A} \to A$ is path equal to $\fix{x}{f\, x}$.
\end{proposition}
\begin{proof}
  Given $p : \Path A ~ a ~ (f \, (\pure{a}))$,
  we proceed by L\"ob induction, i.e., by assuming
  \[
    \term{ih} : \later{(\Path A ~ a ~ (\fix{x}{f\, x}))}.
  \]  
  We define a path 
  \[
    s\defeq \pathlambda{i} f (\pure[\hrt{q \gets \term{ih}}]{q\, i}) 
    ~:~
    \Path A ~ (f (\pure{a})) ~ (f (\pure{\fix{x}{f\, x}})),
  \]
  which is well-typed because the type of the variable $q$ ensures that $q\, 0$ is judgementally equal to $a$, resp. $q\,1$
  and $\fix{x}{f\, x}$.
  Note that we here implicitly use the extensionality principle for later \eqref{eq:later_ext}.
  We compose $s$ with $p$, and then with the inverse of the canonical unfold lemma of Lemma~\ref{prop:unfold-lemma},
  to obtain our path from $a$ to $\fix{x}{f\, x}$.
  We can write out our full proof term, where $p^{-1}$ is the inverse path of $p$, as
  \[
    \fix{\term{ih}}{\pathlambda{i} \comp{j}{A}{[(i=0) \mapsto p^{-1}, (i=1) \mapsto f (\dfix[\Ineg{j}]{x}{f\, x})]}{(f (\pure[\hrt{q \gets \term{ih}}]{q\, i}) )}}.
    \qedhere
  \]
\end{proof}

\subsection{Programming and Proving with Guarded Recursive Types}
In this section we show some simple examples of programming with guarded recursion, and prove properties of our
programs using L\"ob induction and univalence.

\paragraph*{Streams}\label{sec:streams}.
The type of guarded recursive streams in $\gctt$, as with $\gdtt$, are defined as fixed points on the universe:
\[
  \gstream{A} \;\defeq\; \fix{x}{A\times \later[[y\gets x]]{y}}
\]
Note the use of a delayed substitution to transform a term of type $\later{\U}$ to one of type $\U$, as discussed at the start
of Section~\ref{sec:later}. Desugaring to restate this in terms of $\dfixEmp$, we have
\[
  \gstream{A} \;=\;
  A\times\later[[y\gets\dfix[0]{x}{A\times\later[[y\gets x]]{y}}]]{y}
\]
The head function $\hd:\gstream{A}\to A$ is the first projection.
The tail function, however, cannot be the second projection, since this yields a term of type
\begin{equation}\label{eq:streamtail0}
\later[\hrt{y \gets \dfix[0]{x}{A\times\later[\hrt{y\gets x}]{y}}}]{y}
\end{equation}
rather than the desired $\later{\gstream{A}}$.
However we are not far off; $\later{\gstream{A}}$ is judgementally equal to
\begin{align*}
  \later[\hrt{y \gets \dfix[1]{x}{A\times\later[\hrt{y\gets x}]{y}}}]{y},
\end{align*}
which is
the same term as \eqref{eq:streamtail0}, apart from endpoint $1$ replacing $0$.
The canonical unfold lemma (Lemma~\ref{prop:unfold-lemma}) tells us that we can build a path in $\U$ from $\gstream{A}$ to
$A \times \later{\gstream{A}}$; call this path $\abs{i}{\gstream{A}^i}$. Then we can transport between these types:
\[
  \term{unfold} \, s \defeq \term{transp}^i \, \gstream{A}^i \, s
  \qquad\qquad\qquad
  \term{fold} \, s \defeq \term{transp}^i \, \gstream{A}^{\Ineg{i}} \, s
\]
Note that the compositions of these two operations are path equal to identity functions, but not judgementally equal.
We can now obtain the desired tail function $\term{tl} : \gstream{A} \to \later{\gstream{A}}$ by composing the second projection with $\term{unfold}$, so $\term{tl} \, s \defeq (\term{unfold} \, s).2$.
Similarly we can define the stream constructor $\term{cons}$ (usually written infix as $::$) by using $\term{fold}$:
\[
  \term{cons} \defeq \lambda a, s . \term{fold} \, (a, s)
  ~:~
  A \to \later{\gstream{A}} \to \gstream{A}.
\]

We now turn to higher order functions on streams. We define $\term{zipWith} : (A \to B \to C) \to \gstream{A} \to
\gstream{B} \to \gstream{C}$, the stream function which maps a binary function on two input streams to produce an output stream, as
\begin{align*}
  \term{zipWith} \, f \defeq
  \fix{z}{\lambda s_1, s_2 . f \, (\term{hd} \, s_1) \, (\term{hd} \, s_2) \,::\,
  \pure[\vrt{z' \gets z \\ t_1 \gets \term{tl} \, s_1 \\ t_2 \gets \term{tl} \, s_2}]{z' \, t_1 \, t_2}}.
\end{align*}
Of course $\term{zipWith}$ is definable even with simple types and $\later$, but in $\gctt$ we can go further and prove properties
about the function:
\begin{proposition}[$\term{zipWith}$ preserves commutativity]\label{prop:zipwith-preserves-comm}
  If $f : A \to A \to B$ is commutative, then $\term{zipWith}\, f : \gstream{A} \to \gstream{A} \to \gstream{B}$ is commutative.
\end{proposition}
\begin{proof}
  Let $\term{c} : (a_1 : A) \to (a_2 : A) \to \Path B ~ (f\, a_1 \, a_2) ~ (f\, a_2 \, a_1)$ witness commutativity of $f$.
  We proceed by L\"ob induction, i.e., by assuming
  \[
    \term{ih} : \later{\left((s_1:\gstream{A})\to(s_2:\gstream{A})\to \Path B 
      ~ (\term{zipWith}\, f \, s_1 \, s_2)
      ~ (\term{zipWith}\, f \, s_2 \, s_1)\right)}.
  \]
  Let $i:\I$ be a fresh name, and $s_1, s_2 : \gstream{A}$.
  Our aim is to construct a stream which is $\term{zipWith}\, f \, s_1 \, s_2$ when substituting $0$ for $i$, and $\term{zipWith}\, f \, s_2 \, s_1$ when substituting $1$ for $i$.
  An initial attempt at this proof is the term
  \[
    v \,\defeq\,
    \term{c} \, (\term{hd}\, s_1) \, (\term{hd}\, s_2) \, i ~::~ 
    \pure[\vrt{q \gets \term{ih} \\
      t_1 \gets \term{tl}\, s_1 \\
      t_2 \gets \term{tl} \, s_2}]%
    {q \, t_1 \, t_2 \, i}
    ~:~
    \gstream{B},
  \]
  which is equal to 
  \[
    f \, (\term{hd} \, s_1) \, (\term{hd} \, s_2) ~::~
    \pure[\vrt{t_1 \gets \term{tl} \, s_1 \\ t_2 \gets \term{tl} \, s_2}]%
    {\term{zipWith} \, f \, t_1 \, t_2}
  \]
  when substituting $0$ for $i$, which is $\term{zipWith}\, f\, s_1\, s_2$, but \emph{unfolded once}.
  Similarly, $v\subst{1}{i}$ is $\term{zipWith}\, f\, s_2\, s_1$ unfolded once.
  Let $\abs{j}{\term{zipWith}^j}$ be the canonical unfold lemma associated with $\term{zipWith}$ (see Lemma~\ref{prop:unfold-lemma}).
  We can now finish the proof by composing $v$ with (the inverse of) the canonical unfold lemma.
  Diagrammatically, with $i$ along the horizontal axis and $j$ along the vertical:
  \begin{mathpar}
    \begin{tikzcd}[column sep=huge, row sep=huge]
        \term{zipWith}\, f\, s_1\, s_2 \arrow[r,dashed]
        & \term{zipWith}\, f\, s_2\, s_1 \\
        \begin{array}{c}
          f \, (\term{hd} \, s_1) \, (\term{hd} \, s_2) ~:: \\
          \pure[\vrt{t_1 \gets \term{tl} \, s_1 \\ t_2 \gets \term{tl} \, s_2}]%
          {\term{zipWith} \, f \, t_1 \, t_2}
        \end{array}
        \arrow[u,"\term{zipWith}^{\Ineg{j}}\, f\, s_1\, s_2"]
        \arrow[r,"v" below]
        &
        \begin{array}{c}
          f \, (\term{hd} \, s_2) \, (\term{hd} \, s_1) ~:: \\
          \pure[\vrt{t_2 \gets \term{tl} \, s_2 \\ t_1 \gets \term{tl} \, s_1}]%
          {\term{zipWith} \, f \, t_2 \, t_1}
        \end{array}
        \arrow[u,"\term{zipWith}^{\Ineg{j}}\, f\, s_2\, s_1" right]
    \end{tikzcd}
  \end{mathpar}
  The complete proof term, in the language of the implemented type-checker, can be found in Appendix~\ref{app:zipwith}.
\end{proof}

\paragraph{Bisimularity equals equality} Two (guarded) streams are bisimilar when both
their heads and tails are equal. In \gctt\ we can prove that bisimilar streams are
equal, and moreover that the type of bisimilar streams is equal to the type of equal
streams.
\begin{proposition} For all $s,t:\gstream A$, there is a term of type $\Path_{\U}
  (\bisim_A s\, t) (\Path_{\gstream A} s\, t)$.
\end{proposition}
\begin{proof}
We may strengthen
extensionality for later (\autoref{eq:later_ext}), to get that
\[
\Path_{\later A} a b \equiv \later[[ (a'\gets a, b'\gets b ]] \Path_A a'\, b'.\label{eq:later}
\]
This strengthening may be compared to the strong version of functional extensionality which
states an equivalence of the equality type on function types and the type of pointwise equality~\cite[2.9]{hottbook}.

For $s, t : \gstream A$, we have the following chain of equivalences:
\begin{eqnarray*}
\bisim s\, t &\defeq &\Path (\hd s) (\hd t) \times \later [[s' \gets \tl s, t'\gets \tl t]] \bisim s' t'\\
       &\stackrel{by\ ind.}\equiv& \Path (\hd s) (\hd t) \times \later[[s' \gets \tl s, t'\gets \tl t]] \Path s' t'\\
       &\stackrel{(\ref{eq:later})}\equiv& \Path (\hd s) (\hd t) \times \Path (\tl s) (\tl t)\\
       &\equiv &\Path s\, t  
\end{eqnarray*}
The last equivalence is constructed from the $\fold$ and $\unfold$ functions
for streams. 
The statement then follows from univalence.
\end{proof}

\paragraph{Guarded recursive types with negative variance.}
A key feature of guarded recursive types are that they support \emph{negative} occurrences of recursion variables.
This is important for applications to models of program logics~\cite{Birkedal:Step}.
Here we consider a simple example of a negative variance recursive type, namely
$  \term{Rec}_A \defeq \fix{x}{(\later[[x'\gets x]]{x'})\to A} $,
which is path equal to $\later{\term{Rec}_A} \to A$.
As a simple demonstration of the expressiveness we gain from negative guarded recursive types, we define a guarded variant of Curry's Y combinator:
\[
  \begin{array}{lclcl}
    \Delta &\defeq& \lambda x.f(\pure[[x'\gets x]]{((\unfold x') x})) &:& \later\term{Rec}_A\to A \\
    \Ycom &\defeq& \lambda f.\Delta(\pure\fold\Delta) &:& (\later A\to A)\to A,
  \end{array}
\]
where $\term{fold}$ and $\term{unfold}$ are the transports along the path between $\term{Rec}_A$ and $\later{\term{Rec}_A} \to A$.
As with $\term{zipWith}$, $\term{Y}$ can be defined with simple types and $\later$~\cite{Abel:Formalized}; what is new to $\gctt$ is that we can also prove properties about it:

\begin{proposition}[$\term{Y}$ is a guarded fixed-point combinator]
  \label{prop:Y-is-fixed-point-combinator}
  $\term{Y} f$ is path equal to $f \, (\pure({\term{Y} f}))$, for any $f : \later{A} \to A$. Therefore, by Proposition~\ref{prop:unique-fix}, $\term{Y}$ is path equal to $\term{fix}$.
\end{proposition}
\begin{proof}
  $\term{Y} f$ simplifies to $f\, (\pure{(\term{unfold} \, (\term{fold} \Delta) \, (\pure{\term{fold} \Delta}))})$, and $\term{unfold}\, (\term{fold} \Delta)$ is path equal to $\Delta$.
A congruence over this path yields our path between $\term{Y} f$ and
$f (\pure{(\term{Y} f)})$.
\end{proof}


\section{Semantics}
\label{sec:semantics}

In this section we provide sound semantics of $\gctt$, and hence prove the consistency of $\gctt$.
The semantics is based on the category $\totcube$ of presheaves on the category $\cube \times \omega$, where $\cube$ is the \emph{category of cubes}~\cite{Cubical} and $\omega$ is the poset of natural numbers.

Given a countably infinite set of names $i, j, k, \ldots$, the category $\cube$ has as objects finite sets of names $I,J$,
and as morphisms $I \to J$, functions $J \to \dM{I}$, where $\dM{I}$ is the free De Morgan algebra with generators $I$.
Equivalently, the category of cubes is the opposite of the Kleisli category of the free De Morgan algebra monad on finite
sets.
Hence in particular it has products, which are given by disjoint union, a fact used extensively below.

As is standard, contexts of $\gctt$ are interpreted as objects of $\totcube$.
Following the approach of Cohen et al.~\cite{Cubical} types in context $\Gamma$ are interpreted as pairs $(A, c_A)$ of a presheaf $A$ on the category of elements of $\Gamma$ and a chosen \emph{composition} structure $c_A$.
We call such a pair a \emph{fibrant} type.

Semantics of type theory in presheaf categories is well-known.
When interpreting type constructions, such as dependent products, the type part of the pair $(A,c_A)$ is interpreted as usual in presheaf models.
What is new is the addition of composition structure, and much of the work we do in this section is to show that
composition structure is preserved by the various type constructors.
It is complex both to define composition structure, and to show that all types can be equipped with this structure.
To aid with this we describe the composition structure in the internal language of the presheaf topos.
More precisely, in Section~\ref{sec:internal} we use \emph{dependent predicate logic}
extended with four assumptions, of which the most important asserts the existence of an
interval type, as the internal language.
A formulation of compositions in this manner, along with similarly internal descriptions of fillings and faces, appeared (in slightly different form) in an unpublished note by Coquand~\cite{Cubical:internal}.
We recall the precise definitions of these in the following sections, and provide details of some constructions which were omitted in \emph{op. cit}.
The advantage of this approach is that we can show entirely in the internal language that constructions such as dependent
products and sums have compositions satisfying the necessary properties, provided their constituent types do.

Working at this level, the notion of a model of $\ctt$ can be generalised from the category $\widehat{\cube}$ of cubical
sets to any topos whose internal logic satisfies the four assumptions. In particular, these assumptions hold in the presheaf
category $\widehat{\cube\times \DD}$ for any small category $\DD$ with an initial object. The category $\totcube$ is obviously
such a category; we will show that it is one that also allows the constructions of guarded recursion introduced in
Section~\ref{sec:type-theory-examples} to be modelled.

The notion of a model of $\gctt$ is then formulated as follows: a type of $\gctt$ in context $\Gamma$ is interpreted as a
pair of a type $\Lwftype{\Gamma}{A}$ in the internal language of $\totcube$, and a composition structure $c_A$, where $c_A$ is a term in the internal language of a specific type $\LCompTy{\Gamma}{A}$ which we define below after introducing the necessary constructs.
A term of $\gctt$ is then interpreted simply as a term of the internal language.
We use \emph{categories with families}~\cite{dybjer1996internal} as our notion of a model.

This section is organised as follows: Section~\ref{sec:internal} presents the general intermediate language $\cL$ which we use to interpret \gctt\ in.
Section~\ref{sec:model-ctt} models \ctt\ in $\cL$.
Section~\ref{sec:concrete-model-of-L} models $\cL$ in the category of cubical sets.
Section~\ref{sec:model-of-L} considers more general models of $\cL$.
Section~\ref{sec:model-gctt} models $\gctt$ in an extension of $\cL$.
Section~\ref{sec:summary-of-the-model} gives a summary of the semantics.

\subsection{The Dependent Predicate Logic $\cL$}\label{sec:internal}

Instead of formulating our model directly using regular mathematics, we will specify a type-theoretic language $\cL$, tailor-made for the purpose of our model, and inspired by
the internal logic of the presheaf topos of cubical sets, $\widehat{\cube}$.

\begin{figure}
  \begin{align*}
    &\Lwfcxt{\Gamma} && \text{well-formed context} \\
    &\Lwftype{\Gamma}{A} && \text{well-formed type} \\
    &\Lhastype{\Gamma}{t}{A} && \text{typing judgement} \\
    &\Leq{\Gamma}{A}{B} && \text{type equality} \\
    &\Leq{\Gamma}{t}{u}[A] && \text{term equality}
  \end{align*}
  \caption{Judgements of the dependent predicate logic $\cL$.}
  \label{fig:L-judgements}
\end{figure}

$\cL$ is Phoa's \emph{dependent predicate logic}~\cite[Appendix~I]{phoa1992introduction} (see also
Johnstone~\cite[D4.3,4.4]{elephant}) extended with four assumptions, detailed in this section.
Figure~\ref{fig:L-judgements} contains an overview of the types of judgements.
We write $\Prop$ for the type of propositions, $\top$ for true and $\bot$ for false.

In addition to the equality proposition $\LEq{t}{u} : A$, we also have the extensional identity type $\LId{A}{t}{u}$ with equality reflection:
\begin{mathpar}
  \inferrule{
    \Lwftype{\Gamma}{A} \\
    \Lhastype{\Gamma}{t,u}{A}}
  { \Lwftype{\Gamma}{\LId{A}{t}{y}}}
  \and
  \inferrule{
    \Leq{\Gamma}{t}{u}[A]}
  { \Lhastype{\Gamma}{\Lrefl}{\LId{A}{t}{u}}}
  \and
  \inferrule{
    \Lhastype{\Gamma}{p}{\LId{A}{t}{u}}}
  { \Leq{\Gamma}{t}{u}[A]}
\end{mathpar}
$\LIdop$ (the type) and $\cdot = \cdot$ (the proposition) are equally expressive, but for presentation purposes it is practical to have both: Using $\LIdop$ we can easily express the type of \emph{partial elements}
(elements of a type $B$ which are defined only when $t=u$ in $A$) as $\LId{A}{t}{u}\to B$.
Terms of this type, however, are unwieldy to work with since one needs to carry around an explicit equality proof (which will be equal to $\Lrefl$ anyway by the extensionality of the identity type).
Therefore we will implicitly convert back and forth between the type theoretic and the logical representation, and will often
elide proofs, for example writing the context $\Gamma,p:\LId{\Prop}{\phi}{\top}$ as $\Gamma,\phi$.

Following Cohen et al.~\cite{Cubical}, our syntax in Section~\ref{sec:ctt} was \emph{\`a
  la} Russell, i.e.\ it did not contain explicit codes.  The interpretation
in \emph{op. cit.} however contains a special form of Tarski-style universes with an explicit coding function which
commutes with the decoding function $\LElEmp$. These universes can be interpreted in presheaf
models. 
To facilitate the interpretation
of the fibrant universe (in Section~\ref{sec:assumption-3-general}) we assume that our
intermediate language $\cL$ contains an explicit ``elements-of'' operation $\LElEmp$ for a universe $\LU$ of small types.

We now turn to the first of our four assumptions necessary for modelling $\ctt$.

\begin{assumption}[Interval type]
  \label{assumption:interval-type}
  In $\cL$ we have a type $\I$ with
  \begin{mathpar}
    0,1 : \I \and \meet,\join : \I \to \I \to \I \and \Ineg{\cdot} : \I \to \I
  \end{mathpar}
  which is a \emph{De Morgan algebra} which enjoys the \emph{(finitary) disjunction property}:
  \begin{align*}
    0 &\;\;\neq\;\; 1\\
    i \join j = 1 &\implies i = 1 \join j = 1. \tag*{$\blacklozenge$}
  \end{align*}
\end{assumption}
  
\subsubsection{Constructions definable from the interval type}
\label{sec:definable-concepts}

This section will show that the interval type assumption above is sufficient for modelling all of $\ctt$ except for glueing and
the universe, as we can use the interval type to define the face lattice, and hence systems, compositions, fillings, and
paths. While some of the constructions of this section are complex to state, they are mostly fairly obvious translations of the
type-theoretic constructions sketched in Section~\ref{sec:ctt} to the language $\cL$.

We will see three further assumptions, for modelling glueing and the universe, in Section~\ref{subsec:assump_U}.

\paragraph{Faces.}

Using the interval we define the type $\LFace$ as the image of the function $\cdot = 1 : \II \to \Omega$.
More precisely, $\LFace$ is the subset type
\begin{align*}
  \LFace \defeq \left\{ p : \Omega \isetsep \exists (i : \II), p = (i = 1) \right\}
\end{align*}
We will implicitly use the inclusion $\LFace \to \Omega$.
The following lemma in particular states that the inclusion is compatible with all the lattice operations, hence omitting it is unambiguous.
\begin{lemma}\leavevmode
  \begin{itemize}
  \item $\Face$ is a lattice for operations inherited from $\Omega$.
  \item The corestriction $\cdot = 1 : \II \to \Face$ is a lattice homomorphism.
  \item $\Face$ inherits the disjunction property from $\II$.
  \end{itemize}
\end{lemma}

To define partial elements we first define, given a proposition $\hastype{\Gamma}{\phi}{\LFace}$, the subsingleton $\Lface{\phi}$ as
\[
  \Lface{\phi} \defeq \LId{\LFace}{\phi}{\top}.
\]

For this type we have the logical equivalence $\left(\exists! p : [\phi], \top\right) \Leftrightarrow \phi$ which we use below when passing between type-theoretic and logical views in constructions of compositions.

\paragraph{Partial elements.}
Given $\wftype{\Gamma}{A}$ and $\hastype{\Gamma}{\phi}{\LFace}$ we say that a term $t$ is a \emph{partial element} of $A$ of \emph{extent} $\phi$, if $\hastype{\Gamma}{t}{\Pi(p : \Lface{\phi}).A}$.
If we are in a context with $p : \Lface{\phi}$, then we treat such a partial element $t$ as a term of type $A$, leaving implicit the application to the proof $p$, i.e., we write $t$ for $t\,p$.
We similarly will often write $\Gamma, \Lface{\phi}$ for $\Gamma, p : \Lface{\phi}$, and $\Lface{\phi}\to B$ for the
dependent function space $\Pi(p : \Lface{\phi}).B$, leaving the proof variable $p$ implicit.

If we have a term $\Gamma, p:\Lface{\phi} \vdash u : A$ (a partial element), then we define
\begin{align}
  \label{eq:def:elements-extending-u}
  A[\phi \mapsto u] \defeq \Sigma (a:A). \Lfaceto{\left(\LId{A}{a}{u}\right)}{\phi}
\end{align}
as the type of elements of $a$ which equal the partial element $u$ on extent $\phi$.
Note that the second component of the pair is uniquely determined (up to judgemental equality) by equality reflection.
Thus often to construct terms of this type we construct a term of type $A$ and show, in the logic, that it is equal to the partial element $u$ on extent $\phi$.
We do not construct the second component explicitly.

\paragraph{Systems.}
Given $\Lwftype{\Gamma}{A}$, assume we have the following:
\begin{align*}
  \Gamma &\vdash \phi_1, \dots,\phi_n : \LFace \\
  \Gamma &\vdash \phi_1 \join \cdots \join \phi_n = \top \\
  \Gamma, \Lface{\phi_1} &\vdash t_1 : A \\
         & \vdotswithin{\vdash} \\[4pt]
  \Gamma, \Lface{\phi_n} &\vdash t_n : A \\
  \Gamma, \Lface{\phi_i \meet \phi_j} &\vdash t_i = t_j : A, \quad \text{for all $i,j$}.
\end{align*}
In other words: We have $n$ partial elements of $A$ which agree with each other on the intersection of their extents.
We can use the \emph{axiom of definite description} to define the term
\[
  [\phi_1t_1,\dots,\phi_nt_n] \defeq \text{the $x^A$ such that $\chi(x)$}
\]
where
\[
  \chi(x) \defeq (\phi_1 \land (x = t_1)) \lor \cdots \lor (\phi_n \land (x = t_n)).
\]
We call this term a \emph{system}.
The condition for using definite description is a proof (in the logic) of the \emph{unique existence} of such a term.
Given the assumptions above, unique existence of the term follows easily.

Using systems, we generalise the earlier definition \eqref{eq:def:elements-extending-u}:
We define
\[
  A[\phi_1 \mapsto t_1, \dots, \phi_n \mapsto t_n] \defeq A[\phi_1\join\cdots\join\phi_n \mapsto [\phi_1t_1,\dots,\phi_nt_n]],
\]
where the type on the right hand side is using the definition \eqref{eq:def:elements-extending-u}.
Note that $A[\phi \mapsto t]$ is unambiguous, as we have $\Gamma, [\phi] \vdash [\phi t] = t : A$.

\paragraph{Compositions.}
Given $\Lwftype{\Gamma}{A}$, we can define the type of \emph{compositions}:
\begin{align*}
  \LCompTy{\Gamma}{A} \defeq \Pi
  &(\gamma : \I \to \Gamma)\\
  &(\phi : \LFace)\\
  &(u : \Pi (i:\I). \Lface{\phi} \to A(\gamma(i))) . \\
  & A(\gamma(0))[ \phi \mapsto u(0) ] \to A(\gamma(1))[ \phi \mapsto u(1) ].
\end{align*}
Here we treat the context $\Gamma$ as a closed type.
This is justified because there is a canonical bijection between contexts and closed types of the internal language.
The notation $A(\gamma(i))$ means substitution along the (uncurried) $\gamma$, by which we mean the following.
Given some term $\gamma$ of type $\I \to \Gamma$ in some context $\Gamma'$, there is the ``uncurried'' term
$\Lhastype{\Gamma', i : \I}{\gamma(i)}{\Gamma}$ which arises by application of $\gamma$ to $i$.
Finally, we assume the variable $i$ appearing in the type of $u$ is fresh for $\phi$, $\gamma$ and $A$.

Note that there is an important difference between the type of compositions in $\cL$ as defined above and the form of the rule for compositions in $\ctt$.
In the latter the type $A$ depends on $\I$, whereas it seemingly does not in the type of compositions.
This difference however is only superficial since the first argument in the type of compositions is a \emph{path} in $\Gamma$, which gives a dependence of $A$ on $\I$.

Recall that we call a pair of a type $\Gamma \vdash A$ in $\cL$ together with a term $\vdash \bfc : \LCompTy{\Gamma}{A}$ a \emph{fibrant type}.

\paragraph{Fillings.}
Given $\Lwftype{\Gamma}{A}$, we can define the type of \emph{(Kan) fillings}:
\begin{align*}
  \LFillTy{\Gamma}{A} \defeq
  \Pi & (\gamma : \I \to \Gamma) \\
      & (\phi : \LFace) \\
      & (u : \Pi(i: \I) . \Lface{\phi} \to A (\gamma(i))) \\
      & (a_0 : A(\gamma(0))[\phi \mapsto{} u(0)]) \\
      & (i : \I) . \\
      & A(\gamma(i))[\phi \mapsto u(i), (1-i) \mapsto \pi_1 a_0].
\end{align*}
If we have a filling operation $\bff : \LFillTy{\Gamma}{A}$ then we can get a \emph{path lifting} operation which states that given a path $\gamma$ and an element $a_0$ in $A$ \emph{over} $\gamma(0)$ we get a path in $A$ which starts at $a_0$.
Concretely, path lifting is the term $\bfl$ of the following type
\begin{align*}
  \bfl : \Pi&(\gamma : \I \to \Gamma) \\
            &(a_0 : A(\gamma(0))) \\
            &(i:\I). \\
            &A(\gamma(i))[(1-i) \mapsto a_0].
\end{align*}
It is defined as a degenerate case of $\bff$ where $\phi$ is $\bot$, and $u$ therefore is uniquely determined (since it is a partial function defined where $\bot$ holds).
Path lifting is used when constructing compositions for dependent products and sums.

\begin{lemma}[Fillings from compositions]
  If we have a fibrant type $\Lwftype{\Gamma}{A}$ with $\bfc_A : \LCompTy{\Gamma}{A}$, then we have a filling operation $\vdash \bff : \Psi(\Gamma,A)$.
\end{lemma}
\begin{proof}
  We introduce the variables of appropriate types:
  \begin{align*}
    & \gamma : \I \to \Gamma, \\
    & \varphi : \Face, \\
    & u : \Pi(i: \I) . [\phi] \to A (\gamma(i)), \\
    & a_0 : A(\gamma(0))[\phi \mapsto u(0)], \\
    & i : \I.
  \end{align*}
  We need to find a term of type
  \[
    A(\gamma(i))[\phi\mapsto u(i), (i=0) \mapsto \pi_1a_o].
  \]
  We check that the following system is well-defined (in a context with $\phi\join (i=0)$):
  \[
    [\phi u(i\meet j), (i=0)\pi_1 a_0].
  \]
  \begin{itemize}
  \item If $\phi$, then $u(i\meet j) : A(\gamma(i\meet j))$.
  \item If $i=0$, then $\pi_1a_0 : A(\gamma(0)) = A(\gamma(i\meet j))$.
  \item If $\phi$ and $i=0$, then $\pi_1a_0 = u(0) = u(i\meet j)$.
  \end{itemize}
  Note also that this means that
  \[
    A(\gamma(0))[\phi \mapsto u(0)] = A(\gamma(0))[\phi \mapsto u(0),(i=0) \mapsto \pi_1a_0],
  \]
  and therefore we can write the following term:
  \[
    \bfc_A ~ (\lambda j.\gamma(i\meet j)) ~ (\phi \join (i=0))
    ~ (\lambda j. [\phi u(i\meet j), (i=0)\pi_1 a_0]) ~ a_0
  \]
  which has the type
  \[
    A(\gamma(i))[\phi\mapsto u(i), (i=0) \mapsto \pi_1a_o],
  \]
  as was needed.
\end{proof}

\paragraph{Path types.}
Given $\Lwftype{\Gamma}{A}$ and terms $\Lhastype{\Gamma}{t,u}{A}$, we can define the
\emph{Path type}
\[
  \LPath{A}{t}{u} \defeq \Pi(i:\I).A[(1-i)\mapsto t, i \mapsto u]
\]
as the type of paths in $A$, i.e., terms of type $\I \to A$, which start at $t$ and end at $u$.

\subsubsection{Assumptions for glueing and the universe}
\label{subsec:assump_U}

\begin{assumption}[Glueing]
  \label{sec:assumptions:glueing-type}
  There is a type for \emph{glueing} with the following type formation and typing rules
  \begin{mathpar}
    \inferrule{%
      \Lwftype{\Gamma}{A} \\
      \Lwftype{\Gamma, \Lface{\phi}}{T} \\
      \Lhastype{\Gamma, \Lface{\phi}}{f}{T \to A}}{%
      \wftype{\Gamma}{\LGlue{\phi}{T}{f}{A}}} \and \inferrule{%
      \Lhastype{\Gamma}{b}{\LGlue{\phi}{T}{f}{A}}}{%
      \Lhastype{\Gamma}{\Lunglue b}{A[\phi \mapsto f\, b]}} \and \inferrule{%
      \Lhastype{\Gamma, \Lface{\phi}}{f}{T \to A} \\
      \Lhastype{\Gamma, \Lface{\phi}}{t}{T} \\
      \Lhastype{\Gamma}{a}{A[\phi \mapsto f\, t]}}{%
      \Lhastype{\Gamma}{\Lglue{\phi}{t}{a}}{\LGlue{\phi}{T}{f}{A}}}
  \end{mathpar}
  Satisfying the following judgemental equalities:
  \begin{align*}
    \Lglue{1}{t}{a} &= t, \\
    \Lglue{\phi}{b}{(\Lunglue b)} &= b, \\
    \Lunglue (\Lglue{\phi}{t}{a}) & = a. \tag*{$\blacklozenge$}
  \end{align*}
\end{assumption}

The assumption above is essentially the same as the rules for the glueing type in $\ctt$.
One difference is that in the formation rule for $\LGlueEmp$ we do \emph{not} require $f$ to be an equivalence.
We need only additionally assume that $f$ is an equivalence, which is stated in terms of the $\Path$ type, when proving that
glueing is fibrant in Lemma~\ref{lem:glueing-has-composition}.

\begin{assumption}[Fibrant universe]
  \label{sec:assumptions:fibrant-universe}
  There is a \emph{fibrant universe} $\LUf$ which contains pairs of a code in $\LU$
  with an associated composition operator:

  \begin{mathpar}
    \inferrule{%
      \Lhastype{\Gamma}{a}{\LU} \\
      \Lhastype{}{\bfc}{\LCompTy{\Gamma}{\LEl{a}}}}{%
      \Lhastype{\Gamma}{\llparenthesis a, \bfc \rrparenthesis}{\LUf}}
    \and
    \inferrule{%
      \Lhastype{\Gamma}{a}{\LUf}}{%
      \Lwftype{\Gamma}{\LEl{a}}}
    \and
    \inferrule{%
      \Lhastype{\Gamma}{a}{\LUf}}{%
      \Lhastype{}{\LElComp{a}}{\LCompTy{\Gamma}{\LEl{a}}}}
  \end{mathpar}
  satisfying
  \begin{align*}
    \LEl{\llparenthesis a, \bfc \rrparenthesis} &= \LEl{a} \\
    \LElComp{\Gamma}{\llparenthesis a, \bfc \rrparenthesis} &= \bfc \\
    \llparenthesis \LEl{p}, \LElComp{\Gamma}{p} \rrparenthesis &= p. \tag*{$\blacklozenge$}
  \end{align*}
\end{assumption}

\begin{assumption}[$\forall$]
  \label{sec:assumptions:forall}
  We assume that the map $\phi \mapsto \lambda\_.\phi : \LFace \to (\II \to \LFace)$
  between posets has an internal right adjoint $\forall$. Concretely this means that for
  any $\phi : \LFace$ and any $f : \II \to \LFace$ we assume
  \begin{align*}
    \left(\forall (i : \II), \phi \Rightarrow f(i)\right) \Leftrightarrow \left(\phi \Rightarrow \forall(f)\right).
    \tag*{$\blacklozenge$}
  \end{align*}
\end{assumption}

\subsection{A Model of $\ctt$ in fibrant types in $\cL$}
\label{sec:model-ctt}

In this section we show how to use the assumptions from the preceding section to interpret $\ctt$.
In the following sections we show how to extend the interpretation to $\gctt$.
We fix a presheaf category which models $\cL$ and define a \emph{category with families}~\cite{dybjer1996internal} by specifying the type and term functors $\Ty$ and $\Tm$.
The base category of the category with families, the category of contexts, is the chosen presheaf category.
We use the language $\cL$ as the internal language of the presheaf category to describe the objects and morphisms.
Thus to construct the model of $\ctt$ we reuse the types and terms of the language $\cL$, but we only take the \emph{fibrant} types, i.e., the ones with associated composition operators.
The type and term functors are as defined as
\begin{align*}
  \Ty(\Gamma) &\defeq \left\{ ([A],[\bfc_A]) ~\middle|
                \begin{array}{l}
                  \Lwftype{\Gamma}{A} \\
                  \Lhastype{}{\bfc_A}{\LCompTy{\Gamma}{A}}
                \end{array} \right\} \\[1em]
  \Tm(\Gamma,([A], [\bfc_A])) &\defeq \left\{ [t] \mid \Lhastype{\Gamma}{t}{A} \right\}.
\end{align*}
where we use $[A]$ and $[t]$ respectively for the equivalence classes of $A$ and $t$ modulo judgemental equality of $\cL$.
Note that if $A$ and $B$ are equivalent types then $\LCompTy{\Gamma}{A}$ and $\LCompTy{\Gamma}{B}$ are also equivalent, hence the type functor is well-defined.
In constructions and proofs we will omit the mention of equivalence classes and work with representatives.
This is justified since all operations in $\cL$ respect judgemental equality.

Note that the context $\Gamma$ need not correspond to a type, i.e.\ it need not be fibrant.
Context extension and projections can be taken directly from the internal language: $\Gamma.A \defeq \Sigma \Gamma A$, $\mathsf{p} \defeq \pi_1$, and $\mathsf{q} \defeq \pi_2$.

The main challenge addressed in this section is showing that the category with families supports dependent sums, dependent products and universes.
This involves showing that these types of the internal language can be equipped with compositions.
Additionally compositions need to satisfy certain judgemental equalities~\cite[Section $4.5$]{Cubical}.
Checking these equalities is routine from construction of compositions at different types.
Thus we only construct compositions and leave showing judgemental equalities to the reader.

\subsubsection{Interpreting composition}
The following composition term is interpreted in terms of the composition in $\cL$.
\begin{mathpar}
  \inferrule{%
    \hastype{\Gamma}{\phi}{\Face} \\
    \wftype{\Gamma, i:\I}{A} \\
    \hastype{\Gamma, \phi, i : \I}{u}{A} \\
    \hastype{\Gamma}{a_0}{A\subst{0}{i}[\phi\mapsto u\subst{0}{i}]}
  }{%
    \hastype{\Gamma}{\comp{i}{A}{[\phi\mapsto u]}{a_0}}{A\subst{1}{i}[\phi\mapsto u\subst{1}{i}]}
  }.
\end{mathpar}
By assumption we have $c_A$ of type $\LCompTy{\Gamma, i : \I}{A}$ and $u$ and $a_0$ are interpreted as terms in the internal language of the corresponding types.
The interpretation of composition is then the term
\begin{align*}
  \Lhastype{\gamma : \Gamma}{c_A \left(\lambda (i : \I) . (\gamma, i)\right)
                                 \phi
                                 \left(\lambda (i : \I) (p : \Lface{\phi}) . u\right)
                                 a_0}
                            {A(\gamma(1))[ \phi \mapsto u(1)]}
\end{align*}
where we have omitted writing the proof $u(0) = a_0$ on $\Lface{\phi}$.
This proof is constructed from the third premise of the rule.

\subsubsection{Interpreting dependent function types}

Assume that $\interpret{\wftype{\Gamma}{A'}} = (A, \bfc_A)$ and $\interpret{\wftype{\Gamma, x: A'}{B'}} = (B, \bfc_B)$.
We define
\[
  \interpret{\wftype{\Gamma}{(x:A')\to B'}} \defeq (\Pi(x:A).B, \bfc)
\]
where $\bfc_{\Pi(x:A).B} : \LCompTy{\Gamma}{\Pi(x:A).B}$ comes from the following lemma.

\begin{lemma}[\emph{cf.} {\cite[Proposition~$0.3$]{Cubical:internal}}]
  $\Pi$-types preserve compositions: if we have composition terms $\bfc_A : \LCompTy{\Gamma}{A}$ and $\bfc_B : \LCompTy{\Gamma.A}{B}$, then we can form a new composition $\bfc_{\Pi(x:A).B} : \Phi(\Gamma, \Pi(x:A).B)$.
\end{lemma}
\begin{proof}
  Recall that $\Pi$-types commutes with substitution:
  \[
    (\Pi(x:A).B)(\gamma) = \Pi(x:A(\gamma)).B(\gamma),
  \]
  where $B(\gamma)$ is a type in the context with $A$.
  We introduce the variables:
  \begin{align*}
    &\gamma : \I \to \Gamma, \\
    &\phi : \LFace, \\
    &u : \Pi(i : \I) . [\phi] \to \Pi (a:A(\gamma(i))) . B(\gamma(i)), \\
    &c_0 : (\Pi(a : A(\gamma(0))). B(\gamma(0)))[\phi \mapsto u(0)].
  \end{align*}
  We need to find an element in
  \[
    \Pi(a:A(\gamma(1))). B(\gamma(1)),
  \]
  along with a proof that it is $u(1)$ when $\phi=1$.

  Let $a_1 : A(\gamma(1))$ be given. We define $a(i) : A(\gamma(i))[i \mapsto a_1]$ by
  using path lifting on $a_1$, i.e.,
  \[
    a(i) \defeq \bfl ~ (\lambda i . \gamma(1-i)) ~ a_1 ~ (1-i);
  \]
where $\bfl$ is the filling operation defined earlier. Then
  \[
    b_1 \defeq \bfc_B ~(\lambda i . \left< \gamma(i), a(i) \right>) ~ \phi ~ (\lambda i. u(i)(a(i)))
  \]
  will have the type $B(\gamma(1))[\phi \mapsto u(1)a_1]$. So $\lambda a_1. \pi_1 b_1$ has
  the type we are looking for. Now assume $\phi = \top$; then $\lambda a_1 . b_1 = \lambda
  a_1 . u(i)a_1 = u(i)$, which is what we needed.
\end{proof}

\subsubsection{Interpreting dependent sum types}
Dependent sum types $(x : A) \times B$ are interpreted by $\Sigma$-types from $\cL$, along with the composition operation that comes from the following lemma:
\begin{lemma}
  $\Sigma$-types preserve compositions: if we have composition terms $\bfc_A : \LCompTy{\Gamma}{A}$ and $\bfc_B : \LCompTy{\Gamma.A}{B}$, then we can form a new composition $\bfc_{\Sigma(x:A).B} : \Phi(\Gamma, \Sigma(x:A).B)$.
\end{lemma}
The proof proceeds similarly to the previous proof that dependent products have compositions.
 
\subsubsection{Interpreting systems}
We interpret the systems of $\ctt$ by using the systems of $\cL$, and by using the fact that systems preserve compositions: If we have a system $\Lwftype{\Gamma}{[\phi_1 A_1, \dots, \phi_n A_n]}$, then we can define a new composition using a system consisting of the compositions of all the components:
\[
  \bfc \defeq
  \lambda \gamma, \psi, u, a_0 .
  [\phi_1(\gamma \, 1) (\bfc_{A_1} \, \gamma_1 \, \psi \, u \, a_0), \dots, \phi_n(\gamma \, 1) (\bfc_{A_n} \, \gamma_n \, \psi \, u \, a_0)]
  ~ : ~
  \LCompTy{\Gamma}{[\phi_1 A_1, \dots, \phi_n A_n]},
\]
where $\gamma_m : \I \to \Gamma, [\phi_m]$ is the context map $\gamma$ extended with the witness of $[\phi_m]$.

\subsubsection{Interpreting path types}
We interpret the path types:
\[
  \interpret{\Lwftype{\Gamma}{\Path A ~ t~s}} \defeq
  (\LPath{A'}{\interpret{t}}{\interpret{s}}, \bfc),
\]
where $\interpret{A} = (A', \bfc_A)$ and $\bfc : \LCompTy{\Gamma}{\LPath{A'}{\interpret{t}}{\interpret{s}}}$ comes from Lemma~\ref{lem:path-has-comp}.

\begin{lemma}
  \label{lem:path-has-comp}
  Path-types preserve composition: if $\Lwftype{\Gamma}{A}$ is fibrant, then for any $\Lhastype{\Gamma}{t,s}{A}$, we have a composition operator $\bfc : \LCompTy{\Gamma}{\LPath{A}{t}{s}}$.
\end{lemma}
\begin{proof}
  First note that if we have $\Lhastype{\Gamma}{\LPath{A}{t}{s}}$ and $\Lhastype{}{\gamma}{\Gamma}$, then
  \[
    (\LPath{A}{t}{s})(\gamma) = \LPath{A(\gamma)}{t(\gamma)}{s(\gamma)}
    = \Pi (i:\I). A(\gamma)\left[
      \begin{array}{l@{ }l}
        i=0~&\mapsto t(\gamma) \\
        i=1 &\mapsto s(\gamma)
      \end{array}
    \right].
  \]
  Now let
  \begin{align*}
    &\gamma : \I \to \Gamma \\
    &\phi : \LFace \\
    &u : \Pi(j:\I).[\phi] \to \LPath{A(\gamma\, j)}{t(\gamma\, j)}{s(\gamma\, j)} \\
    &p_0 : (\LPath{A(\gamma\, 0)}{t(\gamma\, 0)}{s(\gamma\, 0)})[\phi \mapsto u0]
  \end{align*}
  be given. Our goal is to find a term $p_1$ such that
  \[
    p_1 : (\LPath{A(\gamma\, 1)}{t(\gamma\, 1)}{s(\gamma\, 1)})[\phi \mapsto u1].
  \]
  We will do this by finding a term $q : \Pi(i : \I).A(\gamma\, 1)[\phi \mapsto u\, 1\, i]$, for
  which we verify that $q\, 0 = t(\gamma\, 1)$ and $q1 = s(\gamma\, 1)$, in other words,
  \[
    q : \Pi (i: \I). A(\gamma\, 1)[\phi \mapsto u\, 1\, i, (1-i) \mapsto t(\gamma\, 1), i \mapsto
    s(\gamma\, 1)]
  \]
  as this will be equivalent to having such a $p_1$.

  Let $i:\I$. By leaving some equality proofs implicit we can define the system
  \[
    r(j) \defeq [\phi u\, j\, i, (1-i)t(\gamma\, j), i s(\gamma\, j)]
    : \Pi(j:\I).[\phi \join (1-i) \join i] \to A(\gamma\, j),
  \]
  which is well-defined because $u\, j\, 0=t(\gamma\, j)$ and $u\, j\, 1 = s(\gamma\, j)$. We also have
  that $p_0\, i : A(\gamma\, 0)[\phi \mapsto u\, 0\, i ]$, and since $p_0\, 0 = t(\gamma\, 0)$ and $p_0
  1 = s(\gamma\, 0)$, we can say that
  \[
    p_0\, i : A(\gamma\, 0)[\phi \mapsto u\, 0\, i, (1-i)\mapsto t(\gamma\, 0), i \mapsto s(\gamma\, 0)]
  \]
  so we can use the fibrancy of $A$ to define the term
  \[
    q(i) \defeq \bfc_A \gamma ~ (\phi \join (1-i) \join i) ~ r ~ (p_0\, i)
    : \Pi (i: \I). A(\gamma\, 1)[\phi \mapsto u\, 1\, i, (1-i) \mapsto t(\gamma\, 1), i \mapsto
    s(\gamma\, 1)],
  \]
  which is what we wanted.
\end{proof}

\subsubsection{Interpreting glue types}
We interpret $\Glue$ from $\ctt$ using $\LGlueEmp$ from $\cL$ along with a composition operator, which we have by the following lemma:
\begin{lemma}
  \label{lem:glueing-has-composition}
  Glueing is fibrant, i.e., if we have
  \begin{align*}
    \Gamma &\vdash A \\
    \Gamma &\vdash \phi : \LFace \\
    \Gamma, [\phi] &\vdash T \\
    \Gamma &\vdash w: [\phi] \to T \to A \\
    \Gamma &\vdash p : \operatorname{isEquiv} w
  \end{align*}
  then there is a term $\bfc : \LCompTy{\Gamma}{\LGlue{\phi}{T}{w}{A}}$.
\end{lemma}
The construction of $\bfc$ in the proof of the above lemma is analogous to the construction of the composition operation for glueing in $\ctt$~\cite{Cubical}, but formulated in $\cL$.
A crucial part of the construction is the face $\delta \defeq \forall (\phi \circ \gamma)$, where $\gamma : \I \to \Gamma$, which satisfies that $[\delta]$ implies $[\phi (\gamma \, i)]$ for all $i: \I$.

\subsubsection{Interpreting the universe}
The universe of $\ctt$ is interpreted using the universe of fibrant types $\LUf$.
To define the composition for the universe we follow the construction of Cohen et~al.~\cite{Cubical} in the language $\cL$.

\subsection{A Model of \texorpdfstring{$\cL$}{L} in Cubical Sets}
\label{sec:concrete-model-of-L}
In this section we construct a model of $\cL$ in the category of cubical sets.
Recall that the category of cubes $\cube$ has as objects finite sets of names $i,j,k,\ldots$ and as morphism the functions $J \to \dM{I}$ where $\dM{I}$ is the free De Morgan algebra on $I$.
Alternatively, $\cube$ can be described as the \emph{opposite} of the Kleisli category of the free De Morgan algebra monad on $\finset$.
The category of cubical sets is then the category $\widehat{\cube}$ of presheaves on $\cube$.

In the previous section we showed how to construct a model of $\ctt$ using $\cL$.
Constructing a model of $\cL$ in cubical sets then shows we can give a model of $\ctt$ in cubical sets.
This was shown already by Cohen et~al.~\cite{Cubical}, however we will use results in this section to construct additional models of $\ctt$ in the subsequent section.
In particular, we shall use presheaves over $\cube \times \omega$ to model the full $\gctt$ type theory.

The references in Section~\ref{sec:internal} show how to model dependent predicate logic in any presheaf topos~\cite{phoa1992introduction}, so we omit the verification of this part.
We do however note how the judgements are interpreted since this will be used later in concrete calculations where working in the internal language no longer suffices, e.g., in the definition of the fibrant universe.
\begin{itemize}
\item A context $\Gamma \vdash$ is interpreted as a presheaf.
\item The judgement $\wftype{\Gamma}{A}$ gives a pair of a presheaf $\Gamma$ on $\cube$ and a presheaf $A$ on the category of elements of $\Gamma$.
\item The judgement $\hastype{\Gamma}{t}{A}$ in addition gives a global element of the presheaf $A$.
  Thus for each $I \in \cube$ and $\gamma \in \Gamma(I)$ we have $t_{I,\gamma} \in A(I,\gamma)$ satisfying naturality conditions.
\end{itemize}

Moreover, there is a canonical bijective correspondence between presheaves $\Gamma$ on $\cube$ and interpretations of types $\wftype{\cdot}{\Gamma}$.
This justifies treating contexts as types in $\cL$ when it is convenient to do so.

\subsubsection{The interval type assumption is satisfied}
Take $\II$ to be the functor $y_1$ mapping $I \mapsto \Hom{\cube}{I}{1}=\dM{I}$, where $1$ is the (globally) chosen singleton set.
Since the theory of De Morgan algebras is geometric and for each $I$ we have a De Morgan algebra, together with the fact that the morphisms are De Morgan algebra morphisms, we have that $\II$ is an internal De Morgan algebra, as needed.

Moreover the finitary disjunction property axiom is also geometric, and since it is satisfied by each free De Morgan algebra $\dM{I}$, it also holds internally.

\subsubsection{The glueing assumption is satisfied}\label{sec:assumpt-glue}

We will define glueing internally, apart from a ``strictness'' fix, for which we use the following lemma, which we will also require in Section~\ref{sec:model-gctt}:

\begin{lemma}[Strictification]
  \label{lem:fixing-a-type}
  Let $C$ be a small category and $\top$ a global element\footnote{For a constructive
    meta-theory we add that, for each $c$, equality with $\top_c$ is decidable.} of an object $\KK$ in $\widehat{C}$.
  Denote by $[\phi]$ the identity type $\phi=\top$.

  Let $\Gamma \vdash \phi : \KK$.
  Suppose $\Gamma \vdash T$, $\Gamma, [\phi] \vdash A$ and $\Gamma, [\phi] \vdash T \cong A$ as witnessed by the terms $\alpha, \beta$ satisfying
  \begin{align*}
    \Gamma, [\phi], x : A &\vdash \alpha : T\\
    \Gamma, [\phi], x : T &\vdash \beta : A
  \end{align*}
  plus the equations stating that they are inverses.

  Then there exists a type $\Gamma \vdash \Tl(A,T,\phi)$ such that
  \begin{enumerate}
  \item $\Gamma, [\phi] \vdash \Tl(A,T,\phi) = A$
  \item $\Gamma \vdash T \cong \Tl(A,T,\phi)$ by an isomorphism $\alpha', \beta'$ extending $\alpha$ and $\beta$.
    This means that the following two judgements hold.
    \begin{align*}
      \Gamma, [\phi], x : A &\vdash \alpha = \alpha' : T\\
      \Gamma, [\phi], x : T &\vdash \beta = \beta' : A.
    \end{align*}
    The judgements are well-formed because in context $\Gamma, [\phi]$ the types $\Tl(A, T,\phi)$ and $A$ are equal by the first item of this lemma.
  \item Let $\rho : \Delta \to \Gamma$ be a context morphism. Consider its extension $\Delta, [\phi\rho] \to \Gamma, [\phi]$.
    Then $\Tl(A,T,\phi)\rho = \Tl(A\rho,T\rho,\phi\rho)$.
  \end{enumerate}
\end{lemma}
\begin{proof}
  We write $T'$ for $\Tl(A,T,\phi)$ and define it as follows.
  \begin{align*}
    T'(c,\gamma) =
    \begin{cases}
      A(c,(\gamma, \star)) & \text{ if } \phi_{c,\gamma} = \top_c\\
      T(c, \gamma) & \text{ otherwise}
    \end{cases}
  \end{align*}
  Here $\star$ is the unique proof of $[\phi]$.
  The restrictions are important.
  Given $f : (c, \Gamma(f)(\gamma)) \to \left(d, \gamma\right)$ define $T'(f)$ by cases
  \begin{align*}
    T'(f)(x) &=
        \begin{cases}
          A(f)(x) & \text{ if } \phi_d(\gamma) = \top_d(\star)\\
          \beta_{c,\Gamma(f)(\gamma), \star, T(f)(x)} & \text{ if } \phi_{c,\Gamma(f)(\gamma)} = \top_c\\
          T(f)(x) & \text{ otherwise }
        \end{cases}
  \end{align*}
  We need to check that this definition is functorial.
  The fact that $T'(id) = id$ is trivial.
  Given $f : (d, \Gamma(f)(\gamma)) \to \left(c, \gamma\right)$ and $g : \left(e, \Gamma(f\circ g)(\gamma)\right) \to (d, \Gamma(f)(\gamma))$ we have
  \begin{align*}
    T'(f \circ g)(x)&=
        \begin{cases}
          A(f \circ g)(x) & \text{ if } \phi_{c,\gamma} = \top_c\\
          \beta_{e,\Gamma(f \circ g)(\gamma), \star, T(f \circ g)(x)} & \text{ if } \phi_{e,\Gamma(f \circ g)(\gamma)} = \top_e\\
          T(f \circ g)(x) & \text{ otherwise }
        \end{cases}
  \end{align*}
  In the first and third cases this is easily seen to be the same as $T'(g)(T'(f)(x))$, since if
  $\phi_{e,\Gamma(f \circ g)(\gamma)} \neq \top_e$ then also $\phi_{d,\Gamma(f)(\gamma)} \neq \top_d$ by naturality of $\phi$ and the fact that $\top$ is a global element and the terminal object is a constant presheaf.

  So assume the remaining option is the case, that is, $\phi_{e,\Gamma(f \circ g)(\gamma)} = \top_e$ but $\phi_{c,\gamma} \neq \top_c$.

  We split into two further cases.
  \begin{itemize}
  \item Case $\phi_{d,\Gamma(f)(\gamma)} = \top_d$.
    Then $T'(f)(x) = \beta_{d,\Gamma(f)(\gamma), \star, T(f)(x)}$ and so
    \begin{align*}
      T'(g)(T'(f)(x)) = T'(g)\left(\beta_{d, (\Gamma(f)(\gamma), \star, T(f)(x))}\right)
    \end{align*}
    By naturality of $\beta$ the right-hand side is the same as
    \begin{align*}
      \beta_{e,\Gamma(f\circ g)(\gamma), \star, T(f \circ g)(x)}
    \end{align*}
    which is what is needed.
  \item Case $\phi_{d,\Gamma(f)(\gamma)} \neq \top_d$.
    In this case we have
    \begin{align*}
      T'(f)(x) = T(f)(x)
    \end{align*}
    and
    \begin{align*}
      T'(g)(T'(f)(x)) = \beta_{e,\Gamma(f \circ g)(\gamma), \star, T(g)(T(f)(x))}
    \end{align*}
    which is again, as needed by functoriality of $T$.
  \end{itemize}

  Now, directly from the definition we have the equality $\Gamma, [\phi] \vdash T' = A$.

  It is similarly easy to check the last required property, the naturality of the construction.
  \begin{align*}
    T(A, T, \phi)\rho = T(A\rho,T\rho,\phi\rho).
  \end{align*}

  Finally, we extend the isomorphisms $\alpha$ and $\beta$ to $\alpha'$ and $\beta'$.

  Define $\beta'$ satisfying $\Gamma, x : T \vdash \beta' : T'$ as
  \begin{align*}
    \beta'_{c,\gamma, x} &=
       \begin{cases}
         \beta_{c,\gamma, \star, x} &\text{ if } \phi_c(\gamma) = \top_c(\star)\\
         x & \text{ otherwise }
       \end{cases}
  \end{align*}
  And $\alpha'$ analogously. One needs to check that this is a natural transformation, i.e., a global element.
  Finally, $\beta'$ is the inverse to $\alpha'$ by construction.
\end{proof}

\paragraph{Definition of glueing.}
Given the following types and terms
\begin{align*}
  \Lhastype{\Gamma}{\phi}{\LFace}\\
  \Lwftype{\Gamma,\Lface{\phi}}{T}\\
  \Lwftype{\Gamma}{A}\\
  \Lhastype{\Gamma, \Lface{\phi}}{w}{T \to A}
\end{align*}
we define a new type $\wftype{\Gamma}{\LGlue{\phi}{T}{w}{A}}$ in two steps.

First we define the type\footnote{This type is already present in Kapulkin and Lumsdaine~\cite[Theorem\ 3.4.1]{kapulkin2012simplicial}.}
 $Glue'_\Gamma(\phi, T, A, w)$ in context $\Gamma$ as
\begin{align*}
  Glue'_\Gamma(\phi, T, A, w) = \sum_{a : A}\sum_{t : \Lfaceto{T}{\phi}}\prod_{p : [\phi]} w (t p) = a.
\end{align*}
For this type we have the following property (we write $G'$ for $Glue'(\cdots)$)
\begin{align*}
  \Lwftype{\Gamma, [\phi]}{T \cong G'}
\end{align*}
with the isomorphism consisting of the second projection from right to left and from left to right we use $w$ to construct the pair.

Finally, we define $\LGlue{\phi}{T}{w}{A}$ using Lemma~\ref{lem:fixing-a-type} applied to the type $Glue'$.
Let
\begin{align*}
  \beta : \LGlue{\phi}{T}{w}{A} \to Glue'(\phi, T, A, w)
\end{align*}
be the extension of pairing and
\begin{align*}
  \alpha : Glue'(\phi, T, A, w) \to \LGlue{\phi}{T}{w}{A}
\end{align*}
the extension of the projection as per Lemma~\ref{lem:fixing-a-type}.

Define $\Lunglue : \LGlue{\phi}{T}{w}{A} \to A$ be the composition of $\beta$ and the \emph{first} projection $G' \to A$.
Now if $\phi = \top$ then $\beta$ is just pairing and in this case we also have $\LGlue{\phi}{T}{w}{A} = T$.
So by definition of $G'$ we have $\Lunglue(t) = w t$, validating one of the equalities.

Given $\hastype{\Gamma, [\phi]}{t}{T}$ and $\hastype{\Gamma}{a}{A}$ satisfying $a = w t$ on $[\phi]$ define $\hastype{\Gamma}{\Lglue{\phi}{t}{a}}{\LGlue{\phi}{T}{w}{A}}$ to be pairing followed by $\alpha$.
If $\phi = \top$ we have, because $\alpha$ is just the projection in this case, that $\Lglue{1}{t}{a} = t$.

To appreciate the technicalities in this section, we remark that $Glue'$ is a
pullback. The difference between $Glue'$ and $\Glue$ is that the latter is strict when
pulling back along the identity morphism. Such coherence issues have discussed at length for
the simplicial model; see e.g. Kapulkin and Lumsdaine~\cite{kapulkin2012simplicial}.

\subsubsection{The fibrant universe assumption is satisfied}

This will be proved in greater generality in Section~\ref{sec:assumption-3-general}.

\subsubsection{The $\forall$ assumption is satisfied}

\begin{theorem}
$\psh{\cube}$ models an operation $\forall:\Face^\I\to\Face$ which is right-adjoint to the
constant map of posets $\Face\to\Face^\I$.
\end{theorem}
\begin{proof}
We will first give a concrete description of $\I$ and $\Face$. We know that
$\I(I)=\dM{I}$. We use Birkhoff duality~\cite{birkhoff1937rings} between finite distributive lattices and
finite posets. This duality is given by a functor $J=\Hom{\mathrm{fDL}}-\two$ from finite distributive
lattices to the opposite of the category of finite posets. This functor sends a distributive lattice to its
join-irreducible elements. It's inverse is the functor
$\Hom{\mathrm{poset}}-\two$ which sends a poset to its the distributive lattice of lower sets.
This restricts to a duality between free distributive lattices and powers
of $\two$. A free \emph{De Morgan} algebra on $I$ is a free distributive lattice
on $2I$($=I+I$). We obtain a duality with the category of \emph{even} powers of $\two$
and maps preserving the De Morgan involution~\cite{cornish1977coproducts}.
Moreover, this duality is poset enriched: If $\psi\leq \phi:\dM{I}\to \dM{J}$,
then the corresponding maps on even powers of $\two$, which are defined by pre-composition, are
in the same order relation.

The dual of the inclusion map is the projection $p:\two^{2(I+1)}\to\two^{2I}$. This has a right
adjoint: concatenation with 11: $p\alpha\leq\beta$ iff $\alpha\leq \beta\cdot11$. Concatenation with 11 is natural:
\[\begin{tikzcd}
\two^{2I} \arrow[d,"11"] \arrow{r}{f}& \two^{2J} \arrow[d,"11"]\\
\two^{2(I+1)}\arrow[u,bend left,twoheadrightarrow] \arrow{r}{(f,id)}& \two^{2(J+1)} \arrow[u,bend left,twoheadrightarrow]
\end{tikzcd}\] By duality we obtain a natural right adjoint to the poset-inclusion of DM-algebras.
Finally, we recall that in $\psh{\cube}$ we have $\I^\I(I)=\I(I+1)$ and hence we have an internal map $\forall:\I^\I\to\I$ which is right-adjoint to the constant map $\I\to\I^\I$.

The lattice $\Face$ is the quotient of $\I$ by the relation generated by $x\wedge (1-x)=0$ for all $x$; see~\cite[p7,p17]{Cubical}. Duality turns the quotients into
inclusions. So, we have the inclusion $\{01,10,11\}^I\subset \two^{2I}$ as the set of join
irreducible elements. Here $00$ presents $x\wedge-x$ which is now identified with $\bot$ and
hence no longer join-irreducible. This presentation allows us to define $\forall:\Face^\I\to\Face$. Since
$\Face^\I(I)=\Face(I+1)$, the right adjoint is again given by concatenation by
$11$. We just replace $\two^2$ by $\{01,10,11\}$ in the diagram above.
\end{proof}

\subsubsection{Interpreting base types}

In Section~\ref{sec:internal} we did not provide any means of interpreting base types such as $\Nat$.
In this section we show that the concrete models we are interested in do support that, but we show this (mostly) externally.

A cubical set $A$ is \emph{discrete} if $A \cong \Delta(a)$ for some $a \in \Sets$, where $\Delta : \Sets \to \psh{\cube}$ is the constant presheaf functor.
Equivalently we can characterise discrete types internally, as in Proposition~\ref{prop:internal-discrete-is-external-discrete} below.
This characterisation is useful to define composition for discrete types internally.

\begin{lemma}
  \label{lem:set-of-paths-isomorphism}
  For any cubical set $A$ and any $I \in \cube$ and $i \not\in I$ the function $\beta_I^i : A^\I(I) \to A(I,i)$ defined as
  \begin{align*}
    \beta_I^i(f) = f_{\iota}(i),
  \end{align*}
  where $\iota : I \to I,i$ is the inclusion, is an isomorphism.
  Moreover the family $\beta$ is natural in $I$ and $i$ in the following sense.
  For any $J \in \cube$ and $j \not\in J$ and any $g : I \to J$ we have
  \[
    A(g + (i \mapsto j)) \circ \beta^i_{I} = \beta_J^j \circ A^\I(g).
  \]
\end{lemma}

\begin{corollary}
  \label{cor:iso-implies-injections-isos}
  If the constant map $a \mapsto \lambda \_.a$ of type $A \to A^\I$ is an isomorphism, then $A$ is isomorphic to an object of the form $\Delta(a)$ for some $a \in \Sets$.
\end{corollary}
\begin{proof}  
  Using Lemma~\ref{lem:set-of-paths-isomorphism} we have that for each $I$ and $i\not\in I$, $A(\iota) : A(I) \to A(I,i)$ is an isomorphism, where, again, $\iota$ is the inclusion.
  From this we have that for all $I$, the inclusion $A(\iota_I) : A(\emptyset) \to A(I)$ is an isomorphism.

  Define $a = A(\emptyset)$ and $\alpha : \Delta(a) \to A$ as
  \begin{align*}
    \alpha_I = A(\iota_I).
  \end{align*}

  We then have for any $f : I \to J$ the following
  \begin{align*}
    A(f) \circ \alpha_I = A(f \circ \iota_I) = A(\iota_J).
  \end{align*}
  The latter because $f \circ \iota_I$ and $\iota_J$ are both maps from the empty set, hence they are equal.

  By the previous lemma each $\alpha_I$ is an isomorphism and by the preceding calculation $\alpha$ is a natural transformation.
  Hence $\alpha$ is a natural isomorphism.
\end{proof}

\begin{lemma}
  \label{lem:constant-isomorphism-path}
  If $A$ is isomorphic to $\Delta(a)$ for some $a \in \Sets$ then the obvious morphism $A \to A^\I$ is an isomorphism.
\end{lemma}
\begin{proof}
  The inverse to the isomorphism $\beta$ in Lemma~\ref{lem:set-of-paths-isomorphism} is the morphism $\alpha_I^i$
  \begin{align*}
    \alpha_I^i(a)_f(j) = A([f,(i\mapsto j)])(a).
  \end{align*}
  By assumption $A(\iota)$ for any inclusion $\iota : I \to I,i$ is an isomorphism.
  It is easy to compute that the canonical morphism $A \to A^\I$ arises as the composition of $A(\iota)$ and $\alpha_I^i$.
\end{proof}

\begin{proposition}
  \label{prop:internal-discrete-is-external-discrete}
  Let $A$ be a cubical set.
  The formula
  \begin{align*}
    i : \I, j : \I, f : (\I \to A) \mid \cdot \vdash f(i) = f(j)
  \end{align*}
  holds in the internal language if and only if $A$ is isomorphic to $\Delta(a)$ for some $a \in \Sets$.
\end{proposition}
\begin{proof}
  Suppose the formula holds.
  Then it is easy to see that the constant map from $A$ to $A^\I$ is an isomorphism (the inverse is given, for instance, by evaluation at $0$).
  Corollary~\ref{cor:iso-implies-injections-isos} implies the result.

  Conversely assume $A \cong \Delta(a)$ for some $a \in \Sets$.
  Then by Lemma~\ref{lem:constant-isomorphism-path} the canonical map $\operatorname{const} : A \to A^\I$ is an isomorphism.
  Hence it is internally surjective.
  Thus for any $f : \I \to A$ there is an $a$ in $A$, such that $\operatorname{const} a = f$.
  From this we immediately have $f(i) = f(j)$ for any $i$ and $j$ in $\I$.
\end{proof}

\begin{lemma}
  \label{lem:discrete-types-have-comp}
  Every discrete type $\Lwftype{}{A}$ is fibrant, i.e., it has a composition operator $\bfc_A : \LCompTy{\cdot}{A}$.
\end{lemma}
\begin{proof}
  Since $A$ is discrete, we have that $u(0) = u(1)$ for any $u: \Pi(i:\I).[\phi] \to
  A$. Therefore $A[\phi \mapsto u(0)] = A[\phi \mapsto u(1)]$, so we can choose the
  constant function $\lambda \gamma, \phi, u, a . a$ to be $\bfc_A$, since this will be of
  type $\Phi(\cdot, A)$.
\end{proof}

If we have a composition operator $\bfc_A : \LCompTy{\cdot}{A}$ then we can always construct a weakened version $\bfc_A' : \LCompTy{\Gamma}{A}$ for any $\Gamma$, since $A$ does not depend on $\Gamma$.

Therefore we can interpret the natural number type:
\[
  \interpret{\wftype{\Gamma}{\Nat}} \defeq (\LNat, \bfc_{\LNat}),
\]
where $\bfc_{\LNat}$ is the composition that we get from Lemma~\ref{lem:discrete-types-have-comp}.

\subsection{More General Models of $\cL$}\label{sec:model-of-L}

The type theory $\gctt$ is an extension of $\ctt$, and we intend to model it in the category of presheaves over $\cube \times \omega$.
We first need to establish that we can model $\ctt$ in this category.
This section shows how to do this by demonstrating that we can lift all constructions of $\ctt$ from the category of cubical sets to $\cube \times \DD$, for \emph{any} small category $\DD$ with an initial object.

We first prove some general lemmas.
\begin{lemma}
  \label{lem:inclusion-is-open} Let $\CC,\DD$ be small categories and let $\pi : \CC \times \DD \to \CC$ be the projection functor.
  Then the geometric morphism $\pi^* \dashv \pi_*$ is open.
  If $\DD$ is inhabited then it is also surjective.
\end{lemma}
\begin{proof}
  By Theorem~C.$3.1.7$ of Johnstone~\cite{elephant} it suffices to show that $\pi^*$ is sub-logical.
  To show this we use Lemma~C.$3.1.2$ of \emph{op. cit} (we use notation introduced in that lemma).

  Let $b : \pi(I,n) \to J$ be a morphism in $\CC$.
  Let $U' = (J,n)$, $a = (b, id_n) : (I,n) \to (J,n)$, $r = id_{J} : \pi U' \to J$ and $i = id_J : J \to \pi U'$.
  Then we have $ r \circ i = id_J$ and $i \circ b = \pi a$ as required by Lemma~C.$3.1.2$.

  If $\DD$ is inhabited the projection $\pi$ is surjective on objects, so the corresponding geometric morphism is surjective; see Johnstone~\cite[A4.2.7b]{elephant}
\end{proof}

The above lemma may be read as stating that $\widehat{\CC \times \DD}$ is a conservative extension of $\hat{\CC}$, provided that $\DD$ is inhabited.

\begin{lemma}
  \label{lem:inclusion-is-cartesian-closed}
  If $\DD$ has an initial object 0, then $\pi^*$ is full, faithful, and cartesian closed.
\end{lemma}
\begin{proof}
  The functor $\pi$ has a \emph{left} adjoint, which is the functor
  \begin{align*}
    \iota &: \CC \to \CC \times \DD\\
    \iota(I) &= (I,0)
  \end{align*}
  Trivially we have $\pi \circ \iota = id_\CC$.
  Thus we have that $\iota^*$ is \emph{left} adjoint to $\pi^*$ and because $\pi \circ \iota = id_\CC$ we also have $\iota^* \circ \pi^* = id$ and moreover the counit of the adjunction is the identity.
  Hence the functor $\pi^*$ is full and faithful \cite[Theorem IV.$3.1$]{MacLane:CWM} and by Johnstone~\cite[Corollary A.$1.5.9$]{elephant}, since $\iota^*$ preserves all limits, we have that $\pi^*$ cartesian closed.
\end{proof}

Let $\Omega^\DD$ be the subobject classifier of $\psh{\CC \times \DD}$ and $\Omega$ be
the subobject classifier of $\widehat{\CC}$.

\begin{lemma}
  \label{lem:inclusion-of-subobject-classifiers}
  There is a \emph{monomorphism} $\upsilon : \pi^*\left(\Omega\right) \to \Omega^\DD$ which fits into the pullback
  \begin{displaymath}
    \begin{diagram}
      \pullbacktip \pi^*(1) \ar{r}{\cong} \ar{d}[swap]{\pi^*(\LTrue)} & 1 \ar{d}{\LTrue}\\
      \pi^*\left(\Omega\right) \ar{r}{\upsilon} & \Omega^\DD
    \end{diagram}
  \end{displaymath}
\end{lemma}
\begin{proof}
  As an inverse image, $\pi^*$ preserves monos. So, $\pi^*(\LTrue)$ is a
  mono. Its characteristic map is:
  \begin{align*}
    \upsilon_{I,c}(S) = \left\{ (f,g) \isetsep f \in S \right\}.
  \end{align*}
  This is clearly a mono.
\end{proof}

\begin{corollary}
  \label{cor:equality-factors-through-inclusion}
  If $X = \pi^*(Y)$ then the equality predicate $\chi_\delta : X \times X \to \Omega^\DD$ factors uniquely through $\upsilon$ and the inclusion of the equality predicate of $Y$.
\end{corollary}
\begin{proof}
  The equality predicate is by definition the characteristic map of the diagonal $\delta : X \to X \times X$.
  Let $\delta' : Y \to Y \times Y$ be the diagonal.
  Because $\pi^*$ preserves finite limits the following square is a pullback.
  \begin{displaymath}
    \begin{diagram}
      \pullbacktip
      X \ar{d}[swap]{\delta = \pi^*(\delta')} \ar{r} & \pi^*(1) \ar{d}{\pi^*(\LTrue)} \ar{r}{\cong} \pullbacktip & 1 \ar{d}{\LTrue}\\
      X \times X \ar{r}{\pi^*\left(\chi_{\delta'}\right)} & \pi^*\left(\Omega\right) \ar{r}{\upsilon} & \Omega^\DD
    \end{diagram}
  \end{displaymath}
  and by uniqueness of characteristic maps we have $\upsilon \circ \pi^*\left(\chi_{\delta'}\right) = \chi_{\delta}$.
  Uniqueness of the factorisation follows from the fact that $\upsilon$ is a mono.
\end{proof}
Let $\DD$ be a small category with an initial object.
We show that $\psh{\cube\times\DD}$ models $\cL$.

\subsubsection{The interval type assumption is satisfied}\label{subsub:assum1}
Let $\IDD = \pi^*(\II)$.
Since $\pi^*$ preserves products we can lift all the De Morgan algebra operations of $\II$ to operations on $\IDD$.
The theory of a De Morgan algebra with the finitary disjunction property is geometric~\cite[Section $X.3$]{maclanemoerdijk92}.
Thus the geometric morphism $\pi^* \dashv \pi_*$ preserves validity of all the axioms,
which means that $\IDD$ is an internal De Morgan algebra with the finitary disjunction property.

\paragraph{Faces}
\begin{lemma}
  \label{lem:faces-are-constant}
  Let $\LFaceDD \in \psh{\cube\times\DD}$ and $\LFace \in \widehat{\cube}$ be defined as in Section~\ref{sec:definable-concepts} from $\IDD$ and $\I$.
  Then $\LFaceDD \cong \pi^*(\LFace)$.
\end{lemma}
\begin{proof}
  Let $e : \IDD \to \Omega^\DD$ be the composition $\chi_\delta \circ \langle id, 1 \rangle$ where $\delta$ is the diagonal $\IDD \to \IDD \times \IDD$.
  By definition $\LFaceDD$ is the image of $e$.
  By Corollary~\ref{cor:equality-factors-through-inclusion} and the way we have defined $\IDD$, and all the operations on it, we have that $e = \upsilon \circ \pi^*(e')$ where $e' : \I \to \Omega$ is defined analogously to $e$ above.
  By definition $\LFace$ is the image of $e'$.
  Because inverse images of geometric morphisms preserve image factorisations~\cite[Remark
  1.34]{vickers2007locales}, $\pi^*(\LFace)$ is the image of $\pi^*(e')$. So,
\[
\pi^*\I\twoheadrightarrow\pi^*\Face\rightarrowtail\pi^*\Omega\stackrel\upsilon\rightarrowtail\Omega^\DD\] 
is the unique factorization of the map $[\cdot=1]:\II^\DD\to \Omega^\DD$.
\end{proof}

\subsubsection{The glueing assumption is satisfied}\label{subsub:assum2}

This proceeds exactly as in Section~\ref{sec:assumpt-glue}.

\subsubsection{The fibrant universe assumption is satisfied}
\label{sec:assumption-3-general}

To define the fibrant universe it appears necessary to describe compositions externally.
The following two lemmas aid in this description because they allow us to simplify the exponential $\Gamma^\I$, i.e., the denotation of paths.

\begin{lemma}
  \label{lem:yoneda-and-products}
  Let $\CC$ and $\DD$ be small categories and assume $\CC$ has products.
  Let $k_1 : \CC \to \widehat{\CC}$ and $k_2 : \DD\times\CC \to \widehat{\DD\times\CC}$ be the Yoneda embeddings.
  Let $\pi^* : \widehat{\CC} \to \widehat{\DD\times\CC}$ be the constant presheaf functor.

  For any $d, e \in \CC$ and any $c \in \DD$ there is an isomorphism
  \begin{align*}
    k_2(c, d) \times \pi^*(k_1 e) \cong k_2(c, d \times e)
  \end{align*}
  in $\widehat{\DD\times\CC}$ which is natural in $c$, $d$ and $e$.
\end{lemma}
\begin{proof}
  For any $(c', d') \in \DD \times \CC$
  \begin{align*}
    (k_2(c, d) \times \pi^*(k_1 e))(c', d')
    &= \Hom{\DD\times\CC}{(c', d')}{(c,d)} \times \Hom{\CC}{d'}{e}\\
    &\cong \Hom{\CC}{d'}{d} \times \Hom{\DD}{c'}{c} \times \Hom{\CC}{d'}{e}
      \intertext{and because the hom functor preserves products we have}
    &\cong \Hom{\CC}{d'}{d\times e} \times \Hom{\DD}{c'}{c}\\
    &\cong \Hom{\DD\times\CC}{(c',d')}{(c,d\times e)}\\
    &= k_2(c,d\times e)(c',d')
  \end{align*}
  as required.
\end{proof}

\begin{lemma}
  \label{lem:interval-exponential}
  Let $\DD$ be a small category.
  Let $\IDD \in \psh{\DD\times \cube}$ be the inclusion $\pi^*(\I)$ of $\I \in \psh{\cube}$.
  Let $X \in \psh{\cube \times \DD}$.
  Then for any $c \in \DD$, any $I \in \cube$ and any $i \not\in I$ we have
  \begin{align*}
    X^{\IDD}(I, c) \cong X(I \cup \{i\},c)
  \end{align*}
  naturally in $c$, $I$ and $i$.
\end{lemma}
\begin{proof}
  Using the Yoneda lemma and the defining property of exponents we have
  \begin{align*}
    X^{\IDD}(I, c) &\cong \Hom{\psh{\cube\times\DD}}{y(I,c)}{X^{\IDD}}\\
                  &\cong \Hom{\psh{\cube\times\DD}}{y(I,c)\times \pi^*(\I)}{X}\\
    \intertext{which by Lemma~\ref{lem:yoneda-and-products}, together with the fact that $\I$ is isomorphic to $y\{i\}$, is isomorphic to}
                  &\cong \Hom{\psh{\cube\times\DD}}{y(I\cup\{i\}, c)}{X}\\
                  &\cong X(I \cup \{i\},c).
  \end{align*}
  recalling that disjoint union is the coproduct in the Kleisli category of the free De Morgan algebra monad, and so
  disjoint union defines the product in $\cube$.

  Concretely, the isomorphism $\alpha_{I,i}^c$ maps $\xi \in X^{\IDD}(I, c)$ to $\xi_{(\iota_{I,i},id_{c}, )}(i)$, where $\iota_{I,i} : I \to I,i$ (in $\cube\op$) is the inclusion.
  Its inverse $\beta_{I,i}^c$ maps $x \in X(I \cup \{i\},c)$ to the family of functions $\xi_{(f,g)} : \I(J) \to X(J,d)$ indexed by morphisms $(f,g) : (J,d) \to (I,c)$ (in $(\cube \times \DD)$).
  This family is defined as
  \begin{align*}
    \xi_{(f,g)}(\phi) = X([g, i \mapsto \phi],g)(x)
  \end{align*}
  where $[g, i \mapsto \phi]$ is the map $I, i \to J$ (in $\cube\op$) which maps $i$ to $\phi$ and otherwise acts as $g$.
  This map is well-defined because disjoint union is the coproduct in $\cube\op$.
\end{proof}

\paragraph{Definition of the universe.}
We can now define the universe $\LUfDD$.
For this we assume a Grothendieck universe $\UG$ in our ambient set theory.
First, recall that the Hofmann-Streicher universe $\LUDD$ in $\psh{\cube\times\DD}$ maps $(I,c)$ to the set of functors valued in $\UG$ on the category of elements of $y(I,c)$.
It acts on morphisms $(I,c) \to (J,d)$ by composition (in the same way as substitution in types is modelled).

The elements operation
\begin{align*}
  \frac{\hastype{\Gamma}{a}{\LUDD}}{\wftype{\Gamma}{\Elem(a)}}
\end{align*}
is interpreted as
\begin{align*}
  \Elem(a)((I,c),\gamma) = a_{(I,c),\gamma}(\star)\left(id_{I,c}\right),
\end{align*}
recalling that terms are interpreted as global elements, and $\star$ is the unique inhabitant of the chosen singleton set.

We define $\LUfDD$ analogously to the way it is defined in Section~\ref{sec:concrete-model-of-L}, that is
\begin{align*}
  \LUfDD(I,c) = \Ty(y(I,c)).
\end{align*}
We first look at the following rule.
\begin{mathpar}
  \inferrule{%
    \Lhastype{\Gamma}{a}{\LUDD} \\
    \Lhastype{}{\bfc}{\LCompTy{\Gamma}{\LEl{a}}}}{%
    \Lhastype{\Gamma}{\llparenthesis a, \bfc \rrparenthesis}{\LUfDD}}
\end{mathpar}
Let us write $b = \llparenthesis a, \bfc \rrparenthesis$.
We need to give for each $I \in \cube$, $c \in \DD$ and $\gamma \in \Gamma(I,c)$ a pair
$(b_0, b_1)$ where
\begin{align*}
  &\Lhastype{y(I,c)}{b_0}{\LUDD}\\
  &\Lhastype{\cdot}{b_1}{\LCompTy{y(I,c)}{\LEl{b_0}}}
\end{align*}
Now $b_0$ is easy.
It is simply $a_{(I,c),\gamma}$.
Composition is also conceptually simple, but somewhat difficult to write down precisely.
Elements $\gamma \in \Gamma(I,c)$ are in bijective correspondence (by Yoneda and exponential transpose) to terms $\overline{\gamma}$
\begin{align*}
  \Lhastype{\cdot}{\overline{\gamma}}{{y(I,c)} \to \Gamma}.
\end{align*}
Thus we define
\begin{align*}
  b_1 = \lambda \rho . \bfc \left(\overline{\gamma} \circ \rho\right).
\end{align*}
One checks that this is well-defined and natural by a tedious computation, which we omit here.

We now look at the converse rule in $\mathcal{L}$
\begin{mathpar}
  \inferrule{%
    \Lhastype{\Gamma}{a}{\LUf}}{%
    \Lwftype{\Gamma}{\LEl{a}}} \and \inferrule{%
    \Lhastype{\Gamma}{a}{\LUf}}{%
    \Lhastype{}{\LElComp{a}}{\LCompTy{\Gamma}{\LEl{a}}}}.
\end{mathpar}
To interpret this rule with $\LUfDD$, we interpret for any $a$ and $\bfc$, $\LEl{\llparenthesis a, \bfc
  \rrparenthesis}$ by $\LEl{a}$, where the latter is $\Elem$ map of the Hofmann-Streicher universe.

We need to define $\LElComp{a}$ which we abbreviate to $c$.
We need to give for each $I \in \cube$ and $c \in \DD$ an element $c_{I,c} \in \LCompTy{\Gamma}{\LEl{a}}(I,c)$, and this family needs to be natural in $I$ and $c$.
Given $\gamma \in (\Gamma^{\IDD})(I,c)$ and a fresh $i\not\in I$ we get by Lemma~\ref{lem:interval-exponential} an element $\gamma' \in \Gamma\left((I, i), c\right)$.
Let $\overline{\gamma'} : y((I,i),c) \to \Gamma$ be the morphism corresponding to $\gamma'$ by the Yoneda lemma.
Thus we get from $a$ the term $c'_{I,i,c,\gamma}$
\begin{align*}
  \Lhastype{\cdot}{c'_{I,i,c,\gamma}}{\LCompTy{y((I,i),c)}{\LEl{a}\overline{\gamma'}}}
\end{align*}
and hence by weakening a term
\begin{align*}
  \Lhastype{y(I,c)}{c'_{I,i,c,\gamma}}{\LCompTy{y((I,i),c)}{\LEl{a}\overline{\gamma'}}}
\end{align*}
By Lemma~\ref{lem:yoneda-and-products} and the way $\IDD$ is defined we have a canonical isomorphism $y((I,i),c) \cong y(I,c) \times \IDD$.
We now apply $c'_{I,i,c,\gamma}$ to the path $\delta = \lambda (i : \IDD).(\rho, i)$ to get the element
\begin{align*}
  \Lhastype{\rho : y(I,c)}{c'_{I,i,c,\gamma}\delta}
  {\Pi (\phi : \LFace)(u : \Pi (i:\I). \Lface{\phi} \to B(\delta(i))).B(\delta(0))[ \phi \mapsto u(0) ] \to B(\delta(1))[ \phi \mapsto u(1) ]}
\end{align*}
Where $B = \LEl{a}\overline{\gamma'}$.

From this element we can define $c_{I,c}$ by using the Yoneda lemma again to get the element $\overline{c'_{I,i,c,\gamma}}$ of type
\begin{align*}
  \Pi (\phi : \LFace)(u : \Pi (i:\I). \Lface{\phi} \to B(\delta(i))).B(\delta(0))[ \phi \mapsto u(0) ] \to B(\delta(1))[ \phi \mapsto u(1) ],
\end{align*}
which is a type in context $y(I,c)$, at $(I, c), id_{I,c}$.
To recap, the composition $c$ will map $\gamma \in (\Gamma^{\IDD})(I,c)$ to the element $\overline{c'_{I,i,c,\gamma}}$.

\begin{lemma}
  \label{lem:universe-operations-are-inverses}
  For any $a$ and $\bfc$ of correct types we have
  \begin{align*}
    \LElComp{\llparenthesis a, \bfc \rrparenthesis} &= \bfc\\
    \LEl{\llparenthesis a, \bfc \rrparenthesis} &= \LEl{a}\\
    \llparenthesis \LEl{a}, \LElComp{a} \rrparenthesis &= a
  \end{align*}
\end{lemma}

\subsubsection{The $\forall$ assumption is satisfied}\label{subsub:assum4}

Using Lemmas~\ref{lem:inclusion-is-cartesian-closed} and \ref{lem:faces-are-constant} we can define $\forall$ in $\psh{\cube\times\DD}$ as the inclusion of the $\forall$ from $\widehat{\cube}$.
Lemma~\ref{lem:inclusion-is-open} can then be used to show that the new $\forall$ is the right adjoint to the map $\phi \mapsto \lambda \_.\phi$.

\subsection{A model of \gctt}
\label{sec:model-gctt}

Our construction of a model for $\gctt$ again proceeds via a dependent predicate logic, extending the language $\cL$ used above with counterparts of the later, delayed substitutions, and fixed-point constructs introduced in Sections~\ref{sec:later} and~\ref{sec:fix}.
We call this new language $\cL'$.
One difference between $\gctt$ and $\cL'$ is that in the latter our fixed-point combinator $\Lfix{x}{t}$ has a \emph{judgemental} equality
\begin{align*}
  \Leq{\Gamma}{\Lfix{x}{t}}{t\subst{\Lnext \Lfix{x}{t}}{x}}.
\end{align*}
The $\gctt$ term $\dfix[r]{x}{t}$ is interpreted as $\Lnext(\fix{x}{t})$, forgetting $r$.
This is consistent with the motivation for annotating $\dfix[r]{x}{t}$ with an interval element $r$: it is needed to ensure termination of fixed-point unfolding, but it is semantically irrelevant.

Since $\cL'$ is an extension of $\cL$ we can use it to construct a model of $\ctt$.
The interpretation of delayed fixed point combinator and delayed substitutions of $\gctt$ is straightforward in terms of corresponding constructs of $\cL'$.
The most difficult part is showing that the $\later$ type-former, with delayed substitutions, has compositions, which we do in Section~\ref{sec:interp_later}.
The rest of the section is devoted to providing a model of $\cL'$ in the presheaf category $\psh{\cube\times\omega}$.
Because of the results of the previous subsection this is immediately a model of $\cL$; we need only show that the category $\psh{\cube\times\omega}$ also models the constructs of guarded recursive types.
The constructions are straightforward modifications of constructions used to model guarded recursive types in the topos of trees~\cite{Birkedal+:topos-of-trees,Bizjak-et-al:GDTT}, which is the category $\psh{\omega}$.

\subsubsection{The functor \texorpdfstring{$\Llater$}{``later''}}

In this section we sketch how to model the later type, delayed substitutions, and the fixed-point operator of $\cL'$.
Since the constructions are straightforward modifications of constructions explained in previous work we omit most proofs.
They are, \emph{mutatis mutandis}, as in previous work.

The $\Llater$ functor on the topos of trees $\widehat{\omega}$ was defined by Birkedal et al.~\cite{Birkedal+:topos-of-trees}.
It is straightforward to extend this to the category $\widehat{\cube \times \omega}$, simply ignoring the cube component: given $X \in \widehat{\cube \times \omega}$ define
\begin{align*}
  \Llater{X}(I,n) =
  \begin{cases}
    \{\star\} & \text{ if } n = 0\\
    X(I,m) & \text{ if } n = m + 1
  \end{cases}
\end{align*}
with restrictions inherited from $X$; i.e. if $(f, n \leq m) : (I,n) \to (J,m)$ then
\begin{align*}
  \Llater{X}(f,n \leq m) &: X(J,m) \to X(I,n)\\
  \Llater{X}(f,n \leq m) &=
  \begin{cases}
    ! & \text{ if } n = 1\\
    X(f, k \leq m - 1) & \text{ if } n = k + 1
  \end{cases}
\end{align*}
where $n \leq m$ is the unique morphism $n \to m$ (and similarly $k \leq m - 1$), and $!$ is the unique morphism into $\{\star\}$, the chosen singleton set.

Less concretely, the $\later$ functor on $\widehat{\omega}$ arises via a geometric morphism induced by the successor functor on $\omega$~\cite[Section 2.2]{Birkedal+:topos-of-trees}; the functor above arises similarly from the successor functor on $\cube\times\omega$ which is the identity on the cube component.

There is a natural transformation
\begin{align*}
  \Lnext &: id_{\widehat{\cube \times \omega}} \to \Llater\\
  \left(\Lnext_X\right)_{I,0} &=~!\\
  \left(\Lnext_X\right)_{I,n+1} &= X\left(id_I,(n \leq n+1)\right)
\end{align*}
and a natural family of morphisms $\circledast : \later(Y^X) \times \later X \to \later Y$ making the triple $(\later,\Lnext, \circledast)$ an applicative functor~\cite{McBride:Applicative}.

\begin{lemma}
  \label{lem:model-of-guarded-recursive-terms}
  For any $X$ and any morphism $\alpha : \Llater X \to X$ there exists a unique global element $\beta : 1 \to X$ such that
  \begin{align*}
    \alpha \circ \Lnext \circ \beta = \beta.
  \end{align*}

  Hence the triple $(\widehat{\cube \times \omega},\Llater, \Lnext)$ is a model of guarded recursive terms~\cite[Definition $6.1$]{Birkedal+:topos-of-trees}.
\end{lemma}
\begin{proof}
  Any global element $\beta$ satisfying the fixed-point equation must satisfy the following two equations
  \begin{align*}
    \beta_{I,0}(\star) &= \alpha_{I,0}(\star)\\
    \beta_{I,n+1}(\star) &= \alpha_{I,n+1}\left(\beta_{I,n}(\star)\right).
  \end{align*}
  Hence define $\beta$ recursively on $n$.
  It is then easy to see that $\beta$ is a global element and that it satisfies the fixed-point equation and that it is unique such.
\end{proof}

By Lemma~\ref{lem:model-of-guarded-recursive-terms} and Birkedal et
al.~\cite[Theorem $6.3$]{Birkedal+:topos-of-trees}, $\Llater$ extends to all slices of $\widehat{\cube \times \omega}$, and contractive morphisms on slices have unique fixed-points.

The above translations from $\widehat{\omega}$ to $\widehat{\cube \times \omega}$ are straightforward, but are not
sufficient. First, we need to consider coherence issues, which are ignored by Birkedal et
al.~\cite{Birkedal+:topos-of-trees}. Second, we need to consider delayed substitutions, which we do below,
following the development for $\gdtt$~\cite{Bizjak-et-al:GDTT}. Third, we need to show that the later types are fibrant, i.e.
support the notion of composition, which we do in Section~\ref{sec:interp_later}.

\paragraph{Delayed substitutions}
Semantically a \emph{delayed} substitution of $\cL'$
\begin{align*}
  \dsubst{\xi}{\Gamma}{\Gamma'}
\end{align*}
will be interpreted~\cite{Bizjak-et-al:GDTT} as a morphism $\den{\xi} : \den{\Gamma} \to \Llater\den{\Gamma,\Gamma'}$ making the following diagram commute
\begin{displaymath}
  \begin{diagram}
    & \Llater\den{\Gamma,\Gamma'} \ar{d}{\Llater{\pi}}\\
    \den{\Gamma} \ar{r}[swap]{\Lnext} \ar{ur}{\den{\xi}} & \Llater\den{\Gamma}.
  \end{diagram}
\end{displaymath}
Here $\pi : \den{\Gamma,\Gamma'} \to \den{\Gamma}$ is the composition of projections of the form $\den{\Gamma,\Gamma'',x : A} \to \den{\Gamma,\Gamma''}$.

In particular, if $\Gamma'$ is the empty context then $\pi = id_{\den{\Gamma}}$ and so $\den{\cdot} = \Lnext$, where $\cdot$ is the empty delayed substitution.

Thus given a delayed substitution $\dsubst{\xi}{\Gamma}{\Gamma'}$ and a type
\begin{align*}
  \Lwftype{\Gamma,\Gamma'}{A}
\end{align*}
define
\begin{align*}
  \Lwftype{\Gamma}{\later[\xi]{A}}
\end{align*}
to be
\begin{align*}
  \left(\later[\xi]{A}\right)(I,n,\gamma) =
  \begin{cases}
    1 & \text{ if } n = 0\\
    A\left(I, m, \den{\xi}_{I,n}(\gamma)\right) & \text{ if } n = m + 1
  \end{cases}
\end{align*}
Note that this is exactly like substitution $A\xi$, except in the case where $n = 0$.

In turn, we interpret the rules
\begin{mathpar}
  \inferrule{%
    \wfcxt{\Gamma}}{%
    \dsubst{\cdot}{\Gamma}{\cdot}}
  \and
  \inferrule{%
    \dsubst{\xi}{\Gamma}{\Gamma'} \\
    \hastype{\Gamma}{t}{\later[\xi]{A}}}{%
    \dsubst{\xi\hrt{x \gets t}}{\Gamma}{\Gamma', x:A}}
\end{mathpar}
as follows.
First, the empty delayed substitution is interpreted as $\Lnext$, as we already remarked above.
Given $\dsubst{\xi}{\Gamma}{\Gamma'}$ and $\hastype{\Gamma}{t}{\later[\xi]{A}}$ define
\begin{align*}
  \den{\dsubst{\xi\hrt{x \gets t}}{\Gamma}{\Gamma', x:A}}_{I,n}(\gamma) =
  \begin{cases}
    \star & \text{ if } n = 0\\
    \left(\xi_{I,n}(\gamma), t_{I,n,\gamma}(\star)\right) & \text{ otherwise }
  \end{cases}
\end{align*}

\paragraph{Next} The term-level counterpart is interpreted similarly.
To interpret the rule
\begin{mathpar}
  \inferrule{%
    \hastype{\Gamma,\Gamma'}{t}{A} \\
    \dsubst{\xi}{\Gamma}{\Gamma'}}{%
    \hastype{\Gamma}{\pure[\xi]{t}}{\later[\xi]{A}}}
\end{mathpar}
we proceed as follows.
Given a term $t$ and a delayed substitution $\xi$ we define the interpretation of
\begin{align*}
  \den{\pure[\xi]{t}}_{I,n,\gamma}(\star) =
  \begin{cases}
    \star & \text{ if } n = 0\\
    t_{I,m,\den{\xi}_{I,n}(\gamma)}(\star) & \text{ if } n = m + 1
  \end{cases}
\end{align*}
The type and term equalities for delayed substitutions then follow as in previous work.

\subsubsection{Dependent products, later, and ``constant'' types.}
To define composition for the $\Llater$ type we will need type isomorphisms commuting $\Llater$ and dependent products in certain cases.
We start with a definition.

\begin{definition}
  \label{def:constant-types}
  A type $\Lwftype{\Gamma}{A}$ is \emph{constant with respect to $\omega$} if for all $I \in \cube, n \in \omega, \gamma \in \Gamma(I,n)$ and for all $m \leq n$ the restriction
  \begin{align*}
    A(I,n,\gamma) \to A\left(I,m,\Gamma(id_I,m \leq n)(\gamma)\right)
  \end{align*}
  is the \emph{identity function}\footnote{A perhaps more natural definition would require this function to be a bijection.
    However since this is a technical definition used only in this section we state it only in the generality we need.}
  (in particular, the two sets are equal).
\end{definition}
Note that this is a direct generalisation of ``being constant'' (being in the image of $\pi^*$) for presheaves (i.e., closed types).
Below we will use the shorter notation $\restr{\gamma}{m}$ for $\Gamma(id_I, m\leq n)(\gamma)$.
We have the following easy, but important, lemma.
\begin{lemma}
  \label{lem:constant-closed-under-substitution}
  Being constant with respect to $\omega$ is closed under substitution.
  If $\Lwftype{\Gamma}{A}$ is constant and $\rho : \Gamma' \to \Gamma$ is a context morphism then
  $\Lwftype{\Gamma'}{A\rho}$ is constant.
\end{lemma}

\begin{lemma}
  \label{lem:identity-on-constant-type-is-constant}
  Let $X$ be a presheaf in the \emph{essential image} of $\pi^*$.
  The identity type $\Lwftype{x : X, y : X}{\LId{X}{x}{y}}$ is \emph{constant} with respect to $\omega$.
\end{lemma}
\begin{proof}
  Recall that we have for $\gamma, \gamma' \in X(I,n)$.
  \begin{align*}
    (\LId{X}{x}{y})(I,n,\gamma,\gamma') &=
    \begin{cases}
      \{\star\} & \text{ if } \gamma = \gamma'\\
      \emptyset & \text{ otherwise }
    \end{cases}
  \end{align*}
  Thus for any $m \leq n$
  \begin{align*}
  (\LId{X}{x}{y})(I,m,\restr{\gamma}{m},\restr{\gamma'}{m}) &=
    \begin{cases}
      \{\star\} & \text{ if } \restr{\gamma}{m} = \restr{\gamma'}{m}\\
      \emptyset & \text{ otherwise }
    \end{cases}
  \end{align*}
  But since $\restr{\cdot}{m}$ is an isomorphism we have $\restr{\gamma}{m} = \restr{\gamma'}{m}$ if and only if $\gamma = \gamma'$, which concludes the proof.
  Since all the sets are chosen singletons or the empty set the relevant restrictions are then trivially identity functions.
\end{proof}

Using the assumptions stated above we have the following proposition.
\begin{proposition}
  \label{prop:type-iso-pi-later-constant}
  Assume
  \begin{align*}
    &\Lwftype{\Gamma}{A}\\
    &\Lwftype{\Gamma,\Gamma',x : A}{B}\\
    &\dsubst{\xi}{\Gamma}{\Gamma'}
  \end{align*}
  and further that $A$ is \emph{constant with respect to $\omega$}.

  The canonical morphism from left to right in
  \begin{align}
    \label{eq:type-iso-pi-later-constant}
    \Gamma \vdash \later[\xi]{\Pi(x : A).B} \cong \Pi(x : A).\later[\xi]B
  \end{align}
  is an isomorphism.
  The canonical morphism is derived from the term $\lambda f . \lambda x . \pure[\hrt{\xi,f' \gets f}](f'\,x)$.
\end{proposition}
\begin{proof}
  We need to establish an isomorphism of two presheaves on the category of elements of $\Gamma$.
  Since we already have one of the directions we will first define the other direction explicitly.
  We define
  \begin{align*}
    F : \Pi(x : A).\later[\xi]B \to \later[\xi]{\Pi(x : A).B}.
  \end{align*}
  Let $I \in \cube$, $n \in \omega$ and $\gamma \in \Gamma(I,n)$.
  Take $\alpha \in \left(\Pi(x : A).\later[\xi]{B}\right)(I,n,\gamma)$.
  If $n = 0$ then we have only one choice.
  \begin{align*}
    F_{I,0,\gamma}(\alpha) = \star
  \end{align*}
  So assume that $n = m + 1$.
  Then we need to provide an element of
  \begin{align*}
    F_{I,n,\gamma}(\alpha) \in \left(\Pi(x : A).B\right)\left(I,m,\xi_{I,n}(\gamma)\right).
  \end{align*}
  Which means that for each $f : J \to I$ and each $k \leq m$ we need to give a dependent function
  \begin{align*}
    \beta_{f,k} : (a \in A\left(J,k,(\Gamma,\Gamma')(f,k\leq m)\left(\xi_{I,n}(\gamma)\right)\right)) \to
                 B\left(J,k,(\Gamma,\Gamma')(f,k\leq m)\left(\xi_{I,n}(\gamma)\right), a\right)
  \end{align*}
  Because $\Lwftype{\Gamma}{A}$ we have
  \begin{align*}
    A\left(J,k,(\Gamma,\Gamma')(f,k\leq m)\left(\xi_{I,n}(\gamma)\right)\right)
    =
    A\left(J,k,\pi_{J,k}\left((\Gamma,\Gamma')(f,k\leq m)\left(\xi_{I,n}(\gamma)\right)\right)\right)
  \end{align*}
  where $\pi : \Gamma,\Gamma' \to \Gamma$ is the composition of projections.
  By naturality we have
  \begin{align*}
    \pi_{J,k}\left((\Gamma,\Gamma')(f,k\leq m)\left(\xi_{I,n}(\gamma)\right)\right)
    =
    \Gamma(f,k\leq m)\left(\pi_{I,m}\left(\xi_{I,n}(\gamma)\right)\right).
  \end{align*}
  Now $\pi_{I,m} = \Llater(\pi)_{I,n}$ and so we have (because $\xi$ is a delayed substitution)
  \begin{align*}
    \pi_{I,m}\left(\xi_{I,n}(\gamma)\right)
    =
    \Lnext(\gamma)_{I,n} = \Gamma(id_I,m \leq n)(\gamma).
  \end{align*}
  Hence we have
  \begin{align*}
    A\left(J,k,(\Gamma,\Gamma')(f,k\leq m)\left(\xi_{I,n}(\gamma)\right)\right)
    =
    A\left(J,k,\Gamma(f, k \leq n)(\gamma)\right).
  \end{align*}
  And because $A$ is \emph{constant} we further have
  \begin{align*}
    A\left(J,k,\Gamma(f, k \leq n)(\gamma)\right) = A(J,k+1,\Gamma(f,k+1 \leq n)(\gamma))
  \end{align*}
  (by assumption $k \leq m$ and $n = m + 1$.

  Now $\alpha_{f,k+1}$ is a dependent function
  \begin{align*}
    (a \in A(J,k+1,\Gamma(f,k+1\leq n)(\gamma))) \to
    (\later[\xi]{B})(J,k+1,\Gamma(f,k+1\leq n)(\gamma), a)
  \end{align*}
  And we have
  \begin{align*}
    (\later[\xi]{B})\left(J,k+1,\Gamma(f,k+1\leq n)(\gamma), a\right)
    = B\left(J,k,\xi_{J,k+1}(\Gamma(f,k+1\leq n)(\gamma)), a\right)
  \end{align*}
  (because the relevant restriction of $A$ is the identity).
  Now
  \begin{align*}
    \xi_{J,k+1}(\Gamma(f,k+1\leq n)) &= (\Llater(\Gamma,\Gamma'))(f,k+1\leq n)(\xi_{I,n}(\gamma))\\
                                     &= (\Gamma,\Gamma')(f,k\leq m)(\xi_{I,n}(\gamma)).
  \end{align*}

  Thus, we can define
  \begin{align*}
    \beta_{f,k} = \alpha_{f,k+1}.
  \end{align*}

  The fact that $\beta$ is a natural family follows from the fact that $\alpha$ is a natural family.
  Naturality of $F$ follows easily by the fact that restrictions for $\Pi$ types are defined by precomposition.

  The fact that it is the inverse to the canonical morphism follows by a tedious computation.
\end{proof}

\begin{corollary}
  \label{cor:later-pi-and-faces}
  If $\Lhastype{\Gamma}{\phi}{\LFace}$ then we have an isomorphism of types
  \begin{align}
    \Gamma \vdash \later[\xi]{\Pi(p : \Lface{\phi}).B} \cong \Pi(x : \Lface{\phi}).\later[\xi]{B}.
  \end{align}
\end{corollary}
\begin{proof}
  Using Proposition~\ref{prop:type-iso-pi-later-constant} it suffices to show that $\Lwftype{\Gamma}{\Lface{\phi}}$ is constant with respect to $\omega$.
  Using Lemmas~\ref{lem:constant-closed-under-substitution} and~\ref{lem:identity-on-constant-type-is-constant} it further suffices to show that the presheaf $\LFace$ is in the essential image of $\pi^*$, which is exactly what Lemma~\ref{lem:faces-are-constant} states.
\end{proof}

Finally we need the following technical construction, allowing us to view delayed substitutions as terms in a certain way.
This is needed in showing that later types have compositions in the following section.
\paragraph{Delayed substitutions and later.}
As we mentioned above a delayed substitution $\xi$ is a morphism
\begin{align*}
  \Gamma \to \Llater(\Gamma,\Gamma').
\end{align*}
Hence we can treat it as a term of type $\Llater(\Gamma,\Gamma')$ in context $\Gamma$.
Further given a morphism $\gamma : \Iw \to \Gamma$ we can form the morphism
\begin{align*}
  \xi \circ \gamma : \Iw \to \Llater(\Gamma,\Gamma').
\end{align*}
Finally by using Proposition~\ref{prop:type-iso-pi-later-constant} we can transport $\xi \circ \gamma : \Iw \to \Llater(\Gamma,\Gamma')$ to a term
\begin{align*}
  \overline{\xi \circ \gamma} : \Llater(\Iw \to \Gamma,\Gamma')
\end{align*}
in the empty context.
For this term we have the following equality.

\begin{lemma}
  \label{lem:later-shenanigans}
  Given $\gamma$ and $\xi$ as above then for any type $\Lwftype{\Gamma,\Gamma'}{A}$ we have the equality of types
  \begin{align*}
    i : \Iw \vdash \later[\hrt{\gamma' \gets \overline{\xi \circ \gamma}}]A\left(\gamma'(i)\right) = \later[\xi\gamma(i)]A\left(\gamma(i)\right).
  \end{align*}
  Here $\xi\gamma(i)$ is the \emph{delayed substitution} $\dsubst{}{\Iw}{\Gamma,\Gamma'}$ obtained by substitution in terms of $\xi$.
\end{lemma}
\begin{proof}
  Proof by computation; we
  require the unfolding of the definition of the isomorphism in Proposition~\ref{prop:type-iso-pi-later-constant}.
\end{proof}

\subsubsection{Interpreting later types}\label{sec:interp_later}
The type part of the delayed substitution type is interpreted using delayed substitutions in the language $\cL'$.
In this section we show that we can also construct a composition term for this type.
\begin{lemma}
  Formation of $\later{\xi}$-types preserves compositions.
  More precisely, if $\later[\xi]{A}$ is a well-formed type in context $\Gamma$ and we have a composition term $\bfc_A : \LCompTy{\Gamma,\Gamma'}{A}$, then there is a composition term $\bfc : \LCompTy{\Gamma}{\later[\xi]{A}}$.
\end{lemma}
\begin{proof}
  We introduce the following variables:
  \begin{align*}
    \gamma &: \I \to \Gamma \\
    \phi &: \Face \\
    u &: \Pi(i:\I).\left((\later[\xi]{A}){(\gamma\, i)}\right)^{\phi} \\
    a_0 &: (\later[\xi]{A})(\gamma\, 0)[\phi \mapsto u\, 0].
  \end{align*}
Using Lemma~\ref{lem:later-shenanigans} we can rewrite the types of $u$ and $a_0$:
\begin{align*}
  u &: \Pi(i : \I).\left(\later[\hrt{\gamma' \gets \overline{\xi\circ\gamma}}]{A(\gamma' \, i)}\right)^{\phi} \\
  a_0 &: \later[\hrt{\gamma' \gets \overline{\xi\circ\gamma}}]{A(\gamma'\, 0)}.
\end{align*}
Furthermore, we have the following type isomorphisms:
\begin{align*}
  \Pi(i : \I).\left(\later[\hrt{\gamma' \gets \overline{\xi\circ\gamma}}]{A(\gamma' \, i)}\right)^{\phi}
  &\cong
    \Pi(i : \I). \later[\hrt{\gamma' \gets \overline{\xi\circ\gamma}}]{\left(A(\gamma' \, i)\right)^{\phi}}
    \tag{Corollary~\ref{cor:later-pi-and-faces}} \\
  &\cong
    \later[\hrt{\gamma' \gets \overline{\xi\circ\gamma}}]{\Pi(i : \I). \left(A(\gamma' \, i)\right)^{\phi}},
    \tag{Proposition~\ref{prop:type-iso-pi-later-constant}}
\end{align*}
which means that we have a term
\[
  \tilde{u} : \later[\hrt{\gamma' \gets \overline{\xi\circ\gamma}}]{\Pi(i : \I). \left(A(\gamma' \, i)\right)^{\phi}}.
\]
We can now -- almost -- form the term
\begin{equation}
  \tag{$*$}
  \label{eq:comp-for-later}
  \pure[\vrt{\gamma' \gets \overline{\xi\circ\gamma} \\ u' \gets \tilde{u} \\ a_0' \gets a_0}]{%
    \bfc_A \, \gamma' \, \phi \, u' \, a_0'}
  ~:~ \later[\hrt{\gamma' \gets \overline{\xi\circ\gamma}}]{A(\gamma' \, 1)}.
\end{equation}
In order for the composition sub-term to be well-typed, we need that $a_0' = u\, 0$ under the assumption $\phi$. This is equivalent to saying that the type
\[
  \later[\vrt{\gamma' \gets \overline{\xi\circ\gamma} \\ u' \gets \tilde{u} \\ a_0' \gets a_0}]{%
    (\LId{}{a_0'}{u' \, 0})^{\phi}}
\]
is inhabited. We transform the type as follows:
\begin{align*}
  \later[\vrt{\gamma' \gets \overline{\xi\circ\gamma} \\ u' \gets \tilde{u} \\ a_0' \gets a_0}]{%
  (\LId{}{a_0'}{u' \, 0})^{\phi}}
& \cong
  \left( \later[\vrt{\gamma' \gets \overline{\xi\circ\gamma} \\ u' \gets \tilde{u} \\ a_0' \gets a_0}]{%
  \LId{}{a_0'}{u' \, 0}}\right)^{\phi}
  \tag{Corollary~\ref{cor:later-pi-and-faces}}  \\
& =
  \left( \LId{}{\pure[\vrt{\gamma' \gets \overline{\xi\circ\gamma} \\ u' \gets \tilde{u} \\ a_0' \gets a_0}]{a_0'}}{%
  \pure[\vrt{\gamma' \gets \overline{\xi\circ\gamma} \\ u' \gets \tilde{u} \\ a_0' \gets a_0}]{u'\, 0}} \right)^{\phi} \\
& =
  \left(\LId{}{a_0}{u\, 0}\right)^{\phi},
\end{align*}
where the last equality uses that $\tilde{u}$ is defined using the inverse of $\lambda f \lambda x . \pure[\xi\hrt{f' \gets f}]{f' \, x}$ (Proposition~\ref{prop:type-iso-pi-later-constant}).
By assumption it is the case that $\left(\LId{}{a_0}{u\, 0}\right)^{\phi}$ is inhabited, and therefore (\ref{eq:comp-for-later}) is well-defined.
This concludes the existence part of the proof, as
\[
  \later[\hrt{\gamma' \gets \overline{\xi\circ\gamma}}]{A(\gamma' \, 1)}
  =
  (\later[\xi]{A})(\gamma \, 1),
\]
by Lemma~\ref{lem:later-shenanigans}.

We now have to show that the term \eqref{eq:comp-for-later} is equal to $u\, 1$ under the assumption of $\phi$.
Assuming $\phi$, we get by the properties of $\bfc_A$ that
\[
  \pure[\vrt{\gamma' \gets \overline{\xi\circ\gamma} \\ u' \gets \tilde{u} \\ a_0' \gets a_0}]{%
    \bfc_A \, \gamma' \, \phi \, u' \, a_0'}
  =
  \pure[\vrt{\gamma' \gets \overline{\xi\circ\gamma} \\ u' \gets \tilde{u} \\ a_0' \gets a_0}]{%
    u' \, 1},
\]
and by the definition of $\tilde{u}$ (Proposition~\ref{prop:type-iso-pi-later-constant}) we have that
\[
  \pure[\vrt{\gamma' \gets \overline{\xi\circ\gamma} \\ u' \gets \tilde{u} \\ a_0' \gets a_0}]{u'\,1}
  = u \, 1
\]
as desired.
\end{proof}

Note that in the lemma above we do not require that the types in $\Gamma'$ are fibrant.

\subsection{Summary of the semantics of \gctt}
\label{sec:summary-of-the-model}

The interpretation of the syntax of $\gctt$ follows the pattern for interpreting dependent type theory outlined in
Cohen et al.~\cite[sec~8.2]{Cubical}.
In summary, the following judgements need to be interpreted.
\begin{itemize}
\item $\den{\wfcxt{\Gamma}}$
\item $\den{\wftype{\Gamma}{A}}$
\item $\den{\hastype{\Gamma}{t}{A}}$
\item $\den{\dsubst{\xi}{\Gamma}{\Gamma'}}$
\item $\den{\rho : \Gamma \to \Gamma'}$
\end{itemize}
where the last one is a context morphism.
We have shown the constructions needed to interpret these judgements, but we do not show the details of their interpretations and the verification of the equations.
These follow straightforwardly from the properties of semantic objects we have established.

In summary, the interpretations of the judgements are constructed in three stages.
\begin{enumerate}
\item Every presheaf topos with a non-trivial internal De Morgan algebra $\I$ satisfying
  the disjunction property can be used to give semantics to the subset of the cubical type
  theory $\ctt$ without glueing and the universe. 
  Further, for any category $\DD$, the category of presheaves on $\cube
  \times \DD$ has an interval $\I$, which is the inclusion of the interval in presheaves
  over the category of cubes $\cube$. This was done in Sections~\ref{subsub:assum1} and~\ref{subsub:assum2}.
\item The topos of presheaves $\cube \times \DD$ for any small category $\DD$ with an
  initial object gives a semantics of the entire $\ctt$. This was
  done in Sections~\ref{sec:assumption-3-general} and~\ref{subsub:assum4}.
\item In Section~\ref{sec:model-gctt}, we showed that the category of presheaves on $\cube
  \times \omega$ gives semantics for $\gctt$.
\end{enumerate}%
For all these three cases we have:
\begin{theorem}[Soundness and consistency]
  \label{thm:soundness-consistency}
  The interpretation in particular satisfies the following properties.
  If 
  \begin{align*}
    \eqjudg{\Gamma}{A}{B}
  \end{align*}
  is derivable then the types $\den{\wftype{\Gamma}{A}}$ and $\den{\wftype{\Gamma}{B}}$ are interpreted as the same object.

  If 
  \begin{align*}
    \eqjudg{\Gamma}{t}{s}[A]
  \end{align*}
  is derivable then the terms $\den{\hastype{\Gamma}{t}{A}}$ and $\den{\hastype{\Gamma}{s}{A}}$ are interpreted as equal.

  As a consequence, the judgement $\hastype{}{t}{\Path~\Nat~0~1}$ is not derivable for any closed term $t$.%
\end{theorem}%
This completes the construction of a model of \gctt, as outlined in the beginning of Section~\ref{sec:semantics}.

\section{Conclusion}\label{sec:conclusion-future-work}
In this paper we have made the following contributions:
\begin{itemize}
  \item
We introduce guarded cubical type theory ($\gctt$), which combines features of cubical type theory
($\ctt$) and guarded dependent type theory ($\gdtt$).
The path equality of $\ctt$ is shown to support reasoning about extensional properties of guarded recursive operations,
and we use the interval of $\ctt$ to constrain the unfolding of fixed-points.
\item
We show that $\ctt$ can be modelled in any presheaf topos with an internal non-trivial De Morgan algebra with the disjunction property, glueing, a universe of fibrant types, and an operator $\forall$.
Most of these constructions are done via the internal logic.
We then show that a
class of presheaf models of the form $\psh{\cube \times \DD}$, for any small category $\DD$ with an initial object,
satisfy the above axioms and hence gives rise to a model of $\ctt$.
\item
We give semantics to $\gctt$ in the topos of presheaves over $\cube \times \omega$.
\end{itemize}

\paragraph{Further work.}%
We wish to establish key \emph{syntactic properties} of $\gctt$, namely decidable type-checking and canonicity for base
types. Our prototype implementation establishes some confidence in these properties.

We wish to further extend $\gctt$ with \emph{clock quantification}~\cite{Atkey:Productive}, such as is present in $\gdtt$.
Clock quantification allows for the controlled elimination of the later type-former, and hence the encoding of first-class coinductive types via guarded recursive types.
The generality of our approach to semantics in this paper should allow us to build a model by combining cubical sets with the presheaf model of $\gdtt$ with multiple clocks~\cite{Bizjak-Mogelberg:clock-sync}.
The main challenges lie in ensuring decidable type checking ($\gdtt$ relies on certain rules involving clock quantifiers which seem difficult to implement), and solving the \emph{coherence problem} for clock substitution.

The cubical model is constructive, as indicated, for example, by the forthcoming formalization in
NuPrl\footnote{\url{http://www.nuprl.org/wip/Mathematics/cubical!type!theory/}},
so it is tempting to consider our construction as the interpretation of this model in the internal logic of the topos of trees. 
One technical obstacle
to this is the absence of a constructive development of universes in presheaf
toposes. Hofmann and Streicher~\cite{Hofmann-Streicher:lifting} started from a
Grothendieck universe in a classical set theory, instead of working in the internal logic
of an ambient topos. Moreover, if $\DD$ is an internal category in $\hat\CC$, then
$\hat\CC^\DD\equiv\widehat{\CC\times\DD}$; c.f. Johnstone~\cite[Lem.\ 2.5.3]{elephant}. However,
this is not an isomorphism of categories, so we need to deal
with the usual coherence issues when interpreting type theory. Such obstacles are part of
active research. For example, see work by Voevodsky on building a new theory of models of type
theory~\cite{Voevodsky:MLIdinC}.
Our present theory centers around the geometric morphism $\hat{\pi_1}:\widehat{\cube\times \omega}\to
\hat{\cube}$. This suggests interpreting the topos of trees in the topos of cubical
sets. However, this would not complete the construction of the model, as we would
still need to add the compositions operations.

A related question is how $\gctt$ relates to the model of simplicial
presheaves over $\omega$ in Birkedal et al.~\cite{Mogelberg:2013}. However to answer this, one would
probably first need to understand the precise relation between the (non-guarded) cubical model and the
simplicial model.

Finally, some higher inductive types, like the truncation, can be added to $\ctt$.
We would like to understand how these interact with $\later$.

\paragraph{Related work.}%
Another type theory with a computational interpretation of functional extensionality, but without equality reflection, is
observational type theory ($\ott$)~\cite{Altenkirch:Observational}.
We found $\ctt$'s prototype implementation, its presheaf semantics, and its interval as a tool for controlling unfoldings,
most convenient for developing our combination with $\gdtt$, but extending $\ott$ similarly
would provide an interesting comparison.

Spitters~\cite{Spitters:TYPES} used the interval of the internal logic of cubical sets to
model identity types. Coquand~\cite{Cubical:internal} defined the composition operation internally
to obtain a model of type theory. We have extended both these ideas to a full model of
$\ctt$.
Recent independent work by Orton and Pitts~\cite{Orton:Axioms} axiomatises a model
for $\ctt$ without a universe, again building on Coquand~\cite{Cubical:internal}.
With the exception of the absence of the universe, their development is more general than ours.
Our semantic developments are sufficiently general to support the sound addition of guarded recursive types to $\ctt$.
\paragraph*{Acknowledgements.}

We gratefully acknowledge our discussions with Thierry Coquand, and the comments of our
reviewers of the conference version of this article~\cite{CSL}, and of this article.
This research was supported in part by the ModuRes Sapere Aude Advanced Grant from The Danish Council for Independent Research for the Natural Sciences (FNU),
and in part by the Guarded homotopy type theory project, funded by the Villum Foundation.
Ale\v{s} Bizjak was supported in part by a Microsoft Research PhD grant.

\bibliography{gCTT}

\begin{thebibliography}{10}

\bibitem{Abel:Formalized}
Andreas Abel and Andrea Vezzosi.
\newblock A formalized proof of strong normalization for guarded recursive
  types.
\newblock In {\em {APLAS}}, pages 140--158, 2014.

\bibitem{Altenkirch:Observational}
Thorsten Altenkirch, Conor McBride, and Wouter Swierstra.
\newblock Observational equality, now!
\newblock In {\em PLPV}, pages 57--68, 2007.

\bibitem{Atkey:Productive}
Robert Atkey and Conor McBride.
\newblock Productive coprogramming with guarded recursion.
\newblock In {\em {ICFP}}, pages 197--208, 2013.

\bibitem{CSL}
Lars Birkedal, Ale{\v{s}} Bizjak, Ranald Clouston, Hans~Bugge Grathwohl, Bas
  Spitters, and Andrea Vezzosi.
\newblock {Guarded Cubical Type Theory: Path Equality for Guarded Recursion}.
\newblock In {\em CSL}, 2016.

\bibitem{Mogelberg:2013}
Lars Birkedal and Rasmus~Ejlers M{\o}gelberg.
\newblock Intensional type theory with guarded recursive types qua fixed points
  on universes.
\newblock In {\em LICS}, pages 213--222, 2013.

\bibitem{Birkedal+:topos-of-trees}
Lars Birkedal, Rasmus~Ejlers M{\o}gelberg, Jan Schwinghammer, and Kristian
  St{\o}vring.
\newblock First steps in synthetic guarded domain theory: step-indexing in the
  topos of trees.
\newblock {\em LMCS}, 8(4), 2012.

\bibitem{Birkedal:Step}
Lars Birkedal, Bernhard Reus, Jan Schwinghammer, Kristian St{\o}vring, Jacob
  Thamsborg, and Hongseok Yang.
\newblock Step-indexed {K}ripke models over recursive worlds.
\newblock In {\em POPL}, pages 119--132, 2011.

\bibitem{birkhoff1937rings}
Garrett Birkhoff.
\newblock Rings of sets.
\newblock {\em Duke Mathematical Journal}, 3(3):443--454, 1937.

\bibitem{Bizjak-Mogelberg:clock-sync}
Ale{\v{s}} Bizjak and Rasmus~Ejlers M{\o}gelberg.
\newblock A model of guarded recursion with clock synchronisation.
\newblock In {\em {MFPS}}, pages 83--101, 2015.

\bibitem{Bizjak-et-al:GDTT}
Ale\v{s} Bizjak, Hans~Bugge Grathwohl, Ranald Clouston, Rasmus~Ejlers
  M{\o}gelberg, and Lars Birkedal.
\newblock Guarded dependent type theory with coinductive types.
\newblock In {\em FoSSaCS}, pages 20--35, 2016.

\bibitem{Clouston:Guarded}
Ranald Clouston, Ale\v{s} Bizjak, Hans~Bugge Grathwohl, and Lars Birkedal.
\newblock The guarded lambda-calculus: Programming and reasoning with guarded
  recursion for coinductive types.
\newblock {\em Logical Methods in Computer Science}, 12(3), 2016.

\bibitem{Cubical}
Cyril Cohen, Thierry Coquand, Simon Huber, and Anders {M\"ortberg}.
\newblock Cubical type theory: a constructive interpretation of the univalence
  axiom.
\newblock Unpublished, 2016.

\bibitem{Cubical:internal}
Thierry Coquand.
\newblock Internal version of the uniform {K}an filling condition.
\newblock Unpublished, 2015.
\newblock URL: \url{http://www.cse.chalmers.se/~coquand/shape.pdf}.

\bibitem{cornish1977coproducts}
William Cornish and Peter Fowler.
\newblock Coproducts of de {Morgan} algebras.
\newblock {\em Bulletin of the Australian Mathematical Society}, 16(1):1--13,
  1977.

\bibitem{dybjer1996internal}
Peter Dybjer.
\newblock Internal type theory.
\newblock In {\em TYPES '95}, pages 120--134, 1996.

\bibitem{Hofmann:Extensional}
Martin Hofmann.
\newblock {\em Extensional constructs in intensional type theory}.
\newblock Springer, 1997.

\bibitem{Hofmann-Streicher:lifting}
Martin Hofmann and Thomas Streicher.
\newblock Lifting {G}rothendieck universes.
\newblock Unpublished, 1999.
\newblock URL:
  \url{http://www.mathematik.tu-darmstadt.de/~streicher/NOTES/lift.pdf}.

\bibitem{Huber:canonicity}
Simon Huber.
\newblock {Canonicity for Cubical Type Theory}.
\newblock {\em arXiv:1607.04156}, 2016.

\bibitem{elephant}
Peter~T. Johnstone.
\newblock {\em Sketches of an Elephant: A Topos Theory Compendium}.
\newblock OUP, 2002.

\bibitem{kapulkin2012simplicial}
Chris Kapulkin and Peter~Le{F}anu Lumsdaine.
\newblock The simplicial model of univalent foundations (after {V}oevodsky).
\newblock arXiv:1211.2851, 2012.

\bibitem{MacLane:CWM}
Saunders Mac~Lane.
\newblock {\em Categories for the working mathematician}, volume~5.
\newblock Springer Science \& Business Media, 1978.

\bibitem{maclanemoerdijk92}
Saunders {Mac Lane} and Ieke Moerdijk.
\newblock {\em Sheaves in Geometry and Logic}.
\newblock Springer, 1992.

\bibitem{Martin-Lof-1973}
Per Martin-L{\"o}f.
\newblock An intuitionistic theory of types: predicative part.
\newblock In {\em Logic Colloquium '73}, pages 73--118, 1975.

\bibitem{Coq:manual}
\mbox{The Coq development team}.
\newblock {\em The Coq proof assistant reference manual}.
\newblock LogiCal Project, 2004.
\newblock Version 8.0.

\bibitem{McBride:Applicative}
Conor McBride and Ross Paterson.
\newblock Applicative programming with effects.
\newblock {\em J. Funct. Programming}, 18(1):1--13, 2008.

\bibitem{Mogelberg:tt-productive-coprogramming}
Rasmus~Ejlers M{\o}gelberg.
\newblock A type theory for productive coprogramming via guarded recursion.
\newblock In {\em {CSL-LICS}}, 2014.

\bibitem{Nakano:Modality}
Hiroshi Nakano.
\newblock A modality for recursion.
\newblock In {\em {LICS}}, pages 255--266, 2000.

\bibitem{Norell:thesis}
Ulf Norell.
\newblock {\em Towards a practical programming language based on dependent type
  theory}.
\newblock PhD thesis, Chalmers University of Technology, 2007.

\bibitem{Orton:Axioms}
Ian Orton and Andrew~M. Pitts.
\newblock Axioms for modelling cubical type theory in a topos.
\newblock In {\em CSL}, 2016.

\bibitem{phoa1992introduction}
Wesley Phoa.
\newblock An introduction to fibrations, topos theory, the effective topos and
  modest sets.
\newblock Technical Report ECS-LFCS-92-208, LFCS, University of Edinburgh,
  1992.

\bibitem{Sacchini:Well}
Jorge~Luis Sacchini.
\newblock Well-founded sized types in the calculus of constructions.
\newblock TYPES 2015 talk, 2015.
\newblock URL: \url{http://cs.ioc.ee/types15/programme-slides.html}.

\bibitem{Spitters:TYPES}
Bas Spitters.
\newblock Cubical sets as a classifying topos.
\newblock {\em TYPES}, 2015.

\bibitem{hottbook}
The {Univalent Foundations Program}.
\newblock {\em Homotopy Type Theory: Univalent Foundations for Mathematics}.
\newblock Institute for Advanced Study, 2013.

\bibitem{vickers2007locales}
Steven Vickers.
\newblock Locales and toposes as spaces.
\newblock In {\em Handbook of spatial logics}, pages 429--496. Springer, 2007.

\bibitem{Voevodsky:MLIdinC}
V.~{Voevodsky}.
\newblock {Martin-Lof identity types in the C-systems defined by a universe
  category}.
\newblock {\em ArXiv: 1505.06446}, 2015.

\end{thebibliography}

\newpage
\appendix
\section{\texorpdfstring{$\term{zipWith}$}{zipWith} Preserves Commutativity}
\label{app:zipwith}
We provide a formalisation of Section~\ref{sec:streams} which can be verified by our type-checker.
This file, among other examples, is available in the \texttt{gctt-examples} folder in the type-checker repository.
\lstset{%
literate=     {\\ }{{$\lambda$}}1
              {<}{{$\langle$}}1
              {>}{{$\rangle$}}1
              {|>}{{$\triangleright$}}1
              {<-}{{$\leftarrow$}}1
              {->}{{$\rightarrow$}}1
}
\lstset{
basicstyle=\footnotesize\ttfamily,
morekeywords=[1]{module, data, where},
keywordstyle=[1]{\ttfamily\color{ForestGreen}},
morekeywords=[2]{Id,StrF,Str,StrUnfoldPath,unfoldStr,foldStr,cons,head,tail,zipWithF,zipWith,zipWithUnfoldPath,comm,zipWith_preserves_comm},
keywordstyle=[2]{\ttfamily\color{MidnightBlue}},
comment=[l]{--},
commentstyle={\ttfamily\color{RedOrange}},
}
\begin{lstlisting}
module zipWith_preserves_comm where

Id (A : U) (a0 a1 : A) : U = IdP (<i> A) a0 a1
data nat = Z | S (n : nat)

-- Streams of natural numbers
StrF (S : |> U) : U = (n : nat) * |> [S' <- S] S'

Str : U = fix (StrF Str)

-- The canonical unfold lemma for Str
StrUnfoldPath : Id U Str (StrF (next Str))
 = <i> StrF (dfix U StrF [(i=1)])

unfoldStr (s : Str) : (n : nat) * |> Str
 = transport StrUnfoldPath s

foldStr (s : (n : nat) * |> Str) : Str
 = transport (<i> StrUnfoldPath @ -i) s

cons (n : nat) (s : |> Str) : Str = foldStr (n, s)
head (s : Str) : nat = s.1
tail (s : Str) : |> Str = (unfoldStr s).2

-- Defining zipWith
zipWithF (f : nat -> nat -> nat) (rec : |> (Str -> Str -> Str)) 
 : Str -> Str -> Str
 = (\ (s1 s2 : Str) ->
      (cons (f (head s1) (head s2))
      	    (next [zipWith' <- rec, s1' <- tail s1 , s2' <- tail s2]
		  zipWith' s1' s2')))

zipWith (f : nat -> nat -> nat) : Str -> Str -> Str
 = fix (zipWithF f zipWith)

zipWithUnfoldPath (f : nat -> nat -> nat)
 : Id (Str -> Str -> Str)
      (zipWith f)
      (zipWithF f (next (zipWith f)))
 = <i> zipWithF f (dfix (Str -> Str -> Str) (zipWithF f) [(i=1)])

-- Commutativity property
comm (f : nat -> nat -> nat) : U = (m n : nat) -> Id nat (f m n) (f n m)

-- zipWith preserves commutativity.
zipWith_preserves_comm (f : nat -> nat -> nat) (c : comm f)
 : (s1 s2 : Str) -> Id Str (zipWith f s1 s2) (zipWith f s2 s1)
 = fix
   (\ (s1 s2 : Str) -> 
      <i> comp (<_> Str) 
               (cons (c (head s1) (head s2) @ i)
                     (next [q <- zipWith_preserves_comm
                           ,t1 <- tail s1
                           ,t2 <- tail s2]
                           q t1 t2 @ i))
               [(i=0) -> <j> zipWithUnfoldPath f @ -j s1 s2
               ,(i=1) -> <j> zipWithUnfoldPath f @ -j s2 s1])
\end{lstlisting}

\end{document}